\documentclass[12pt]{amsart}

 \usepackage{amsmath}       
 \usepackage{amsthm}        
 \usepackage{amssymb}       
 \usepackage{empheq}   
 \usepackage[shortlabels]{enumitem}     
 \usepackage{mathrsfs}      
 \usepackage{pstricks}      
 \usepackage{color}        
 \usepackage{slashed}
 \usepackage[all]{xy}       
    \SelectTips{cm}{10}     
    \everyxy={<2.5em,0em>:} 
 \usepackage[colorlinks,
             linkcolor=black,
             pagebackref,
             bookmarksnumbered=true]{hyperref}
             
             \usepackage{cite}

\def\newprooflikeenvironment#1#2#3#4{
      \newenvironment{#1}[1][]{
          \refstepcounter{equation}                                 
          \begin{proof}[{\rm\csname#4\endcsname{#2~\theequation}
          }]
         \it \def\qedsymbol{#3}}
         {\end{proof}}}                                             
\makeatother                                                         
 \newprooflikeenvironment{definition}{Definition}{$\diamond$}{textbf}%
 \newprooflikeenvironment{example}{Example}{$\diamond$}{textbf}     
 \newprooflikeenvironment{remark}{Remark}{$\diamond$}{textbf}       
 \newtheorem{theorem}[equation]{Theorem}                            
 \newtheorem{lemma}[equation]{Lemma}                                
 \newtheorem{proposition}[equation]{Proposition}

\title[A holography theory of PSM and deformation quantization]{A holography theory of Poisson sigma model and deformation quantization}

\author{Xiaoyi Cui}
\address{School of Mathematics (Zhuhai), Sun Yat-sen University, Zhuhai 519082, China}
\email{xiaoyi.cui@gmail.com}
\author{Chenchang Zhu}
\address{Mathematics Institute, Georg-August-University of G\"{o}ttingen, G\"{o}ttingen 37073, Germany}
\address{Center for Mathematical Sciences, Huazhong University of Science and Technology, Wuhan 430074, China}
\email{chenchang.zhu@mathematik.uni-goettingen.de}
\date{\today}

\begin{document}

\thanks{
We thank Si Li for helpful discussion. XC was partially supported by NSFC grant 11801588 and by Guangdong Natural Science Foundation grant 2018A030313273. XC thank Kavli IPMU and the University of G\"{o}ttingen for hospitality, where part of the work has been done. 
}

\maketitle

\begin{abstract} 
We construct a Chern-Simons type of theory using the $l_\infty$ algebra encoded by a Poisson structure on arbitrary Riemann surfaces with boundaries. A deformation quantization within the Batalin-Vilkovisky framework is performed by constructing propagators with Dirichlet boundary condition on Fulton-MacPherson compactified configuration space. Our results show that the BV quantization is independent of several gauge choices in propagators, which leads to global observables that are candidates for geometric invariants of Poisson structure and topological invariants for the worldsheet structure. At the level of local observables, a Swiss-Cheese algebra structure has been identified. If the Poisson structure is symplectic, the two-dimensional theory is homotopic to a boundary theory. This is known in the classical case, and we confirm that the quantum homotopy exists as well.
\end{abstract}

\tableofcontents

\section{Introduction and main results}

It is known that Kontsevich's deformation quantization of Poisson manifold has a field-theoretical interpretation using Poisson sigma models \cite{Kont97, CF00}. At operadic level the classical Poisson structure $\mathbb P_1$ is deformed into a the homotopy associative $\mathbb E_1$ structure in any reasonable deformation quantization scheme. According to the perturbative quantum field theory developed in \cite{CG}, the quantization map should be realised as a weak equivalence from the classical observables of certain one-dimensional field theory, to the quantum observables. Here the classical observable has the structure of a locally constant factorisation $\mathbb P_1$-algebra, while the quantum counter-part has the structure of a locally constant factorisation $\mathbb E_1$-algebra. The quantization map relates the above two structures.  A full description of Fedosov quantization in terms of one-dimensional field theory and the relevant factorisation algebra structure has been studied in \cite{GLL}. However, the general Poisson case remains interesting for us. On the other hand, Tamarkin and Tsygan \cite{TT00} showed that there is an alternative proof of formality by investigating the homotopy $\mathbb E_2$ structure, which should be observed from the the local observable correlation of the underlying two-dimensional topological quantum theory. We re-examine the Cattaneo-Felder model, in the formulation of a deformation problem in a curved $l_\infty$ algebra with a $\mathbb P_2$ structure, over an arbitrary Riemann surface $\Sigma$ with boundaries. The quantum local observable in the bulk exhibits $\mathbb E_2$ multiplication, and then the local boundary observable should exhibits $\mathbb E_1$ multiplication, connecting to the case of Kontsevich. We consider the BV quantization scheme of \cite{CG} as generalised to the boundary case, and seek for possible geometric/topological invariants coming from observables, and then obtain a boundary theory of the Poisson sigma model. The structure of local observables should be understood as a Swiss-Cheese ($SC$-) algebra \cite{Vor97}. This can be understood as one partial result of holograph principal. Within this paper, by full holograph, we mean that in the symplectic case, the deformation problem encoded in our theory is only meaningful at the boundary of the worldsheet $\Sigma$. There are rich holograph phenomena in quantum field theories, whose underlying algebraic structure remains to be uncovered yet.  

A Poisson manifold is a smooth manifold $M$ with a globally defined skew symmetric bivector field $\Pi \in \wedge^2 \mathcal T_M$ subject to the Jacobi identity. We shall see that there exists a $\Omega_M$-linear curved $l_\infty$ algebra structure over the graded vector space $\mathfrak g = \Omega_M\otimes(\mathcal T_M[-1]\oplus \mathcal T^\vee _M)$ with operations $\{l'_n\}_{n\geq0}$, which encodes the infinite jet bundle of polyvector fields $(PV_M, [-,\Pi])$. Using this $l_\infty$ structure, a Chern-Simons type of theory \cite{Co11} is constructed as below, which will again be referred to as a Poisson sigma model by us without confusion.

The space of fields $\mathcal E$ in our model is a formal moduli stack of almost constant maps from the worldsheet $\Sigma$ (possibly with boundaries) to the target curved $l_\infty$ algebra $\mathfrak g[1]$. The mapping space is modelled by the extended $l_\infty$ algebra $\Omega_\Sigma\otimes_{\mathbb R} \mathfrak g$.
There is a degree $-1$ symplectic pairing 
on  $\mathcal E$ by 
$$ (\alpha\otimes a)\otimes (\beta\otimes b)\mapsto \int_\Sigma \alpha\wedge \beta  \langle a, b\rangle_{1},\quad \alpha,\beta\in \Omega_\Sigma,\, a, b\in \mathfrak g,$$
where $\langle-,-\rangle_1$ is induced from the degree $1$ $\Omega_M$-linear pairing between $\Omega_M\otimes\mathcal T_M$ and $\Omega_M\otimes\mathcal T^\vee_M[1]$. 
The BV action is now given by 
$$S^{BV} [X+\eta]= \int_\Sigma \langle X+\eta, \frac12 d(X+\eta)+\sum_{n=0}^\infty \frac1{(n+1)!}l'_n(X+\eta)^{\otimes n}\rangle_1$$
for $X\in \Omega_\Sigma\otimes_{\mathbb R} \Omega_M\otimes_{\mathcal O_M}\mathcal T_M$ and $\eta \in \Omega_\Sigma\otimes_{\mathbb R} \Omega_M\otimes_{\mathcal O_M}\mathcal T^\vee_M[1]$.
As we will see in Sec.~\ref{linfty}, the operations $\{l'_n\}_{n\geq0}$ have two origins, the $\{l_n\}_{n\geq0}$ encoding the formal geometry of $M$, and $\{\Pi_n\}_{n\geq1}$ coming from the Poisson bivector field $\Pi$ part. If $\Pi$ is trivial, the operations $\{l_n\}_{n\geq0}$ alone define a cotangent theory in the sense of \cite{Co11}. It is therefore useful to distinguish those two classes of operations in the theory. Under this notation, we have that\footnote{This definition differs from the previous one up to a total derivative, which does not matter once the boundary condition is chosen.}
\begin{eqnarray*}
S^{BV} [X, \eta]&=& \int_\Sigma\,\Big( \langle \eta, dX+ \sum_{n=0}^\infty \frac1{(n+1)!}l_n(X)^{\otimes n}\rangle_1 \\&&+ \langle \eta, \sum_{n\geq1} \frac1{2!}\frac1{(n-1)!}\Pi_n(\eta, X^{\otimes (n-1)})\rangle_1\Big).
\end{eqnarray*}

The $(-1)$-symplectic pairing on $B\mathfrak g$ corresponds to a homotopy $\mathbb P_0$-structure \cite{CG} at the level of classical global observables. In particular, there is a $\mathbb P_0$ Poisson bracket, i.e., the antibracket, $\{-,-\}_0$ on observables whose singular supports intersect transversely, under which the classical action $S^{BV}$ becomes the Hamiltonian function for the Chevalley-Eilenberg differential (a.k.a., vector field) over the extended curved $l_\infty$ algebra $\Omega_\Sigma\otimes_{\mathbb R}\mathfrak g\equiv \mathcal E[-1]$. The Classical Master Equation is given by the nilpotence of the Chevalley-Eilenberg differential, $\{S^{BV},S^{BV}\}_0=0$.

In the BV quantization, one defines a BV laplacian $\Delta$ of degree $1$ on $\mathcal E$, which is a simple tensor product of analytic part and algebra part. The analytic part is given by a smooth diagonal class $\Xi$ (see Sec.~\ref{BV:prop}) over $\Sigma\times\Sigma$. The algebraic part is the natural pairing $\langle-,-\rangle_1$ on $\mathfrak g$. Overall $\Delta$ provides a bilinear pairing on the linear observables, and extends to a second order differential operator on full observables, whose failure of being a derivation is captured by the quantum bracket $\{-,-\}$, in the sense that 
\[\{a,b\} =\Delta(ab) - \Delta(a)b - (-1)^{|a|}a\Delta(b),\quad \forall a,b\in \mathcal O(\mathcal E).\]

Our main results are as follows. Firstly, we consider the BV quantization of the Poisson model on compact Riemann surface with boundaries, and show that the obstruction to the quantization vanishes. In the process of constructing effective action, one needs to fix a propagator. But by analysing the parameterised/family Quamtum Master Equation (QME), one sees that the gauge choices involved in the computation leads to homotopic result.

\begin{theorem}
The effective action $I_{\rm eff}$ for theory defined by classical action $S^{BV}$ given above satisfies QME
\[ dI_{\rm eff} + \hslash \Delta I_{\rm eff}+\{I_{\rm eff},I_{\rm eff}\}=0.\] 
Further more, different gauge-choices of BV quantization (i.e., gauge-choices of propagators) can be encoded in a parameterised QME, by promoting to the $\Omega_{\mathbb I}^{*}$-valued quantities, the effective action $\tilde I_{\rm eff}$ satisfies
\[ (d_\Sigma+d_t)\tilde I_{\rm eff} + \hslash \tilde \Delta \tilde I_{\rm eff}+\{\tilde I_{\rm eff},\tilde I_{\rm eff}\}\tilde {}=0.\]
\label{thmqme}
\end{theorem}

A classical linear observable supported over $U\subset \Sigma$ is an element in the distributional-valued form $\Omega'_U\otimes \mathfrak g^\vee$. A classical observable of homogeneous degree $n$ supported over $U$ is, loosely speaking, an $n$-folded external product\footnote{Since distributions can not be multiplied, this does not quite make sense. We refer the serious readers to \cite{CG} for a careful definition.} of the linear observables, which, viewed as a distribution, is supported on the Cartesian product $U^{\times n}$. In particular, we are interested in the local observables, which is supported near a single point in $\Sigma$. One can show that such observables, in the classical level, are isomorphic to $C^*(\mathfrak g)$ if $z$ is a bulk point and to $C^*(\mathfrak h)$ if $z$ is a boundary point. Here by $\mathfrak h $ we mean the $l_\infty$ algebra $\Omega_M\otimes \mathcal T_M[-1]$ which encodes the smooth structure of $M$, see Sec.~\ref{linfty}. Formally, for an observable $O_z$ supported near $z$, the quantization gives rise to a non-local observable formally written as 
\[\frac{d}{dt}\big({\rm log}e^{h\partial_P} e^{\frac{I}{\hslash}+tO_z}\big)|_{t=0}.\]
Such formal expression is a concise way of writing down a combinatoric formula involving Feynman graph enumeration \cite{Co08}. It will become clear that in our theory, such graph enumeration is well-defined after we prove a similar proposition about $I_{\rm eff}$ in Sec.~\ref{sec:qme}.
The BV quantization gives rise to the quantum product of local observables, both for bulk observables and for boundary observables, which is a well-defined Swiss-Cheese algebra. This, in a proper sense, shall provide an intuitive understanding on why Tamarkin's proof of formality theorem needs the little 2-disk operad.

\begin{theorem}
For each contractible bulk open set $U\subset \Sigma^\circ$, there exists a homotopy family of multiplication on $C^*(\mathfrak g)[[\hslash]]$ given by \[ {\rm Conf}(2,U) \times C^*(\mathfrak g)\otimes C^*(\mathfrak g)\to C^*(\mathfrak g)[[\hslash]].\] Similarly, for each each contractible $U\subset\Sigma$ at the boundary such that $U\cap\partial\Sigma\neq\emptyset$ is contractible, there exists a homotopy family of multiplication on $C^*(\mathfrak h)[[\hslash]]$ given by \[ {\rm Conf}(2, U\cap \partial\Sigma) \times C^*(\mathfrak h)\otimes C^*(\mathfrak h)\to C^*(\mathfrak h)[[\hslash]].\]
\label{obser-prop}
\end{theorem}

Finally we look at the holograph property of the Poisson model when there is a symplectic structure at present. The holograph is a prediction which can be verified at classical level for the Cattaneo-Felder model, which we shall briefly review by the end of this section. Mathematically we will show that there exists a one-dimensional theory $S^{1d}$ as described in \cite{GLL} such that $S^{1d}$ is homotopic to the Poisson BV theory $S^{BV}$, i.e., there exists a homotopy $H\in Obs(\Sigma)$ such that $S^{BV} - S^{1d} = \{S^{BV}, H\}_0$.
We show that a similar relation holds at the quantum level.

\begin{theorem}
The quantised observable $H_{\rm eff}$ induces a homotopy between quantum Poisson model and the quantum 1d sigma model, i.e., 
\[dH_{\rm eff}+\Delta H_{\rm eff}+\{I_{\rm eff}, H_{\rm eff}\} = S^{1d}_{\rm eff} - I_{\rm eff}.\]
\label{bulkbdy}
\end{theorem}

Holography phenomena were investigated mathematically also in \cite{CMR, MSW19}, where the authors relate the deformation quantization on the bulk to the canonical quantization in the boundary theory. Formulating canonical quantization rigorously needs more involving machinery, which we shall not pursue here. Our results should be understood as an example of simplified version relating deformation quantization of different dimensions.

In the remaining of the introduction section, we briefly review the original physical model. Given a smooth manifold $M$ equipped with Poisson bi-vector field in local coordinates given by $\Pi^{ij}\partial_i\partial_j$, one starts by considering the following fields:
$$X\in C^\infty (\Sigma, M),\quad \eta\in \Omega^1_\Sigma\otimes \Gamma(\Sigma, X^*T_M). $$
Now the action is given by a functional on the space of fields locally written as:
$$S[X, \eta] = \int_\Sigma \eta_i d X^i +\frac12 \Pi^{ij}(X)\eta_i\wedge \eta_j.$$

Indeed, $\eta$ has to transform like $-\Pi^{-1} dX$, as indicated by the equation of motion. In this case, the gauge symmetry is given by the Poisson Lie algebroid action, and a BRST quantization is sufficient to quantize the system. Alternatively, we could {\it integrate out} $\eta$ to get a new action $$S[X] = \int_\Sigma X^* \Pi^{-1},$$
which can be reduced to a one-dimensional theory easily via transgression, using the fact that the symplectic form is locally exact. Mathematically the classical procedure of ``integrating fields out" can be encoded in a homotopy between the bulk and the boundary theory, which we state in Prop.~\ref{classicalbdythy}.

The physical Poisson sigma model has a gauge symmetry, induced by the Hamiltonian vector fields on the Poisson manifold. However, the gauge algebra is not closed, and this is the common place where the original BV formulation could deal with\cite{CF00}. In our approach, however, classical BV theory is viewed as a formal moduli problem, where both the gauge symmetry and the moduli problem are naturally encoded in an $l_\infty$ structure equipped with a $(-1)$-shifted symplectic structure. 

The structure of the paper is as follows. In Sec.~\ref{linfty}, we explain the construction of space of field as an $l_\infty$ algebra. In Sec.~\ref{classicalbvthy} we discuss about the classical BV theory in the boundary case. In Sec.~\ref{quant} we construct the propagator over Fulton-MacPherson compactified configuration space, define the effective quantum action of Poisson model, and show some vanishing results. We also present the proof of Theorem~\ref{thmqme} in Sec.~\ref{sec:qme}. In Sec.~\ref{secobserv} we discuss quantum observables in the theory. In Sec.~\ref{obs:glo}, we show that the global observables lead to a topological invariant for the worldsheet $\Sigma$ and provides useful probes for the Poisson structure. In Sec.~\ref{obs:loc}, we show that the local observables form a Swiss-Cheese algebra, which we show by proving Theorem~\ref{obser-prop}. In  Sec.~\ref{obs:bdy}, a form of holograph of the Poisson sigma model with symplectic target is concluded and proved in Theorem~\ref{bulkbdy}.

\section{Classical BV theory with boundaries}
\subsection{The $l_\infty$ structure}
\label{linfty}

The classical BV formalism, in the approach of \cite{Co08, CG} starts with an $l_\infty$ algebra which controls a deformation problem encodes the classical dynamics of the system. For a Poisson manifold $M$, if one forgets about the Poisson structure, the algebra is given by the following result.

\begin{lemma}[\cite{Co11, GG_ahat}]
Given a smooth real manifold $M$, there exists a contractible family of curved $l_\infty$ algebra structure on $\mathfrak g_M$ such that 
\begin{itemize}
\item as a vector space, $\mathfrak g_M\simeq \Omega_M\otimes \mathcal T[-1]_M$
\item the $l_1$ structure is parameterised by the choices of connections on $T_M$, which results in a contractible of homotopy family of $l_\infty$ structures.
\end{itemize}
The de Rham complex of the cotangent bundle $\Omega_M\otimes \mathcal T^\vee_M$ can be realised as an $l_\infty$ module $\mathfrak g_M^\vee[-1]$ over $\mathfrak g_M$ by the coadjoint representation. The module structure defines a minimally extended $l_\infty$ algebra $\mathfrak g_M\oplus \mathfrak g_M^\vee[-1]$. Moreover, the natural bilinear map 
\[\mathfrak g_M\otimes \mathfrak g^\vee_M\to \Omega_M\] 
now becomes a degree-$(-1)$ invariant pairing\footnote{The definition of invariant pairing over an $l_\infty$ algebra will be given in Sec.~\ref{classicalbvthy}.} for the extended $l_\infty$ algebra.
\end{lemma}

For smooth functions on $ M$, there are two maps 
$j: C^\infty_M \to \mathscr J$ given by sending a function to the flat section of the jet sheaf, and the other way $p: \mathscr J\to C^\infty_M$ given by projecting on the zeroth jet. It is easy to see that $p\circ j=id_{C^\infty_M}$. If we have vector bundles $E\to M$, we need to choose a local trivialisation for $E$, written as $\Phi_U:\pi^{-1} U\subset E\to U\times E_0$ for $U$ an open cover of $M$, then the trivialisation uniquely determines a local splitting of the section sheaf $\mathcal E(U)$ by $C^\infty_U\otimes_{\mathbb R} E_0$, where $E_0$ is the fiber of some point $x_0\in U$. Now we can still apply the map $\Sigma$ on the first entry, and this gives a local splitting of $\mathscr J(\mathcal E)$ via $\sigma\otimes_\mathbb{R} id\circ \Phi\circ p$. On the bundle side, the gluing is given by the transition function $S\in C^\infty_U\otimes {\rm GL} (E_0)$, a compatible gluing on the jet bundle is given by the jets of $S$. In this way we make $\mathscr J(\mathcal E)$ into a $\mathscr J$-module, equipped with a compatible flat connection denoted by $\partial_E$. This construction is well known \cite{BD, BZF}, and has been used in constructing BV theories raised from formal moduli problems \cite{Co11, GG_ahat}. 

To conclude, the construction involves two step: firstly, the infinite jet functor $\mathscr J$ sends each finite dimensional $C^\infty$-vector bundle into finitely generated $\mathscr J(C^\infty)$-modules. At the same time, the differential operators are sent to $C^\infty$-linear maps (i.e., view each $\mathscr J(C^\infty)$-module as an infinite dimensional vector bundle). Secondly, in order to combine the gluing property and the multiplicative structure, one extends the base ring $C^\infty$ by the dg ring $(\Omega^*, d)$. In this light, the Grothendieck connection on the jet bundle extends to a differential on the de Rham complex, and hence the inclusion 
\[C^\infty \to (\Omega\otimes \mathscr J,\partial)\]
is a weak equivalence of filtered differential algebra. 

For each Poisson manifold $M$, the Poisson bivector field $\Pi$ can be viewed as the ``Hamiltonian" function on the shifted cotangent bundle $T^\vee[1]M$ which generates a differential over the polyvector fields, the latter being a cohomologically graded cdga $(PV^\bullet_M, [-,\Pi])$ over $\mathbb R$. To encode this structure in formal geometry, one looks at the de Rham complex of the infinite jets of poly-vector fields. 

The space $\Omega_M\otimes \mathscr J(PV_M)$ is a filtered algebra, whose associated graded objects being a graded commutative algebra over $\Omega_M$. Upon the choice of connection over $T_M$, the $\mathbb R$-linear differential $[-,\Pi]$ over $PV_M$ induces a $C^\infty_M$-linear differential $d_\Pi$ over $\mathscr J(PV_M)$, which in turn, become $\Omega_M$-linear in $\Omega_M\otimes \mathscr J(PV_M)$. 
Further more, as a result of previous analysis, the connection on $T_M$ also identifies $\mathscr J(PV_M)$ with the Chavelley-Eilenberg cochain of $\mathfrak g_M\oplus \mathfrak g_M^\vee[-1]$. So there is an isomorphism between the associated graded cdga of $(\Omega_M\otimes \mathscr J(PV_M), \partial + id\otimes d_\Pi)$ and $C^*(\mathfrak g_M\oplus \mathfrak g_M^\vee[-1])$, the latter with an inherited differential. 

We shall refer to the structures of the cdga and the $l_\infty$ algebra obtained using the information of Poisson bivector field $\Pi$ as the $\Pi$-twisted structure. Indeed, given a Lie algebroid over $M$, there is a way to associate a curved $l_\infty$ algebra, as given by \cite{GG_l}, and the $\Pi$-twisted case is a special example of the more general construction.

\begin{proposition}
Given a smooth Poisson manifold $M$, there is a contractible family of curved $l_\infty$ structure over the dg vector space $\mathfrak g_M\oplus \mathfrak g^\vee_M[-1]$, such that
\begin{itemize}
\item the $l_1$ structure encodes the choice of connections on $T_M$ and the constant part of Poisson bivector field $\Pi$,
\item the $l_n$ structure is given by the n-bracket on $\mathfrak g_M$, the $l_n$-module structure of $ \mathfrak g^\vee_M[-1]$, as well as the geodesic expansion of $\Pi$,
\item there is a pairing of degree $-1$ such that the pairing is symmetric with respect to $l_n$ for each $n\geq1$.
\end{itemize}
\end{proposition}

Note that the shifted symplectic pairing, if understood as a pairing on the space of field $\mathcal E$, has degree $1$, which induces a homotopy $\mathbb P_2$-structure (or, Gerstenharber algebra) on the de Rham complex of $\mathscr J(PV_M)$. Classical Poisson is $\mathbb P_1$, the passing from that to the $\mathbb P_2$ structure in $PV_M$ has been used to define the notion of {\it center} in \cite{Safranov}. In the classical version, the center computes Poisson cohomology for $(M,\Pi)$, and our $l_\infty$ algebra is a formal geometry version.

\begin{proposition}
There exists a contractible family of curved $l_\infty$ algebra structures on $\mathfrak g:=\Omega_M\otimes(\mathcal T_M[-1]\oplus \mathcal T^\vee_M) $, whose Chevalley-Eilenberg cochain $C^*(\mathfrak g)$ computes the Poisson cohomology of $M$.
If the Poisson manifold $M$ is symplectic, the $\Pi$-twisted curved $l_\infty$ algebra   has its cohomology algebra isomorphic to $H^*(M,\mathbb R)$.
\label{targetalgebra}
\end{proposition}

\begin{proof}  We have seen that on $\Omega_M\otimes\mathcal T_M[-1]\equiv \mathfrak h$ there are contractible choices of $l_\infty$ algebra structure, in one-to-one correspondence to the choice of connections on $T_M$. Further more, this determines a connection on $T^\vee_M$, there is an $l_\infty$-module structure on $\mathfrak h^\vee[-1]$. The desirable $l_\infty$ algebra structure on $\mathfrak g\cong \mathfrak h \oplus \mathfrak h^\vee[-1]$ comes from a minimal extension of $\mathfrak h$ by module $\mathfrak h^\vee[-1]$, further twisted by the Poisson structure. I.e., there is a component of $l_1$ operation given by the anchor map $\mathfrak h^\vee[-1]\to \mathfrak h$.

The result about $C^*(\mathfrak g) \cong (\Omega_M\otimes \mathscr J(PV), \nabla+ \{\mathscr J(\Pi),-\})$ is sheaf-theoretical, hence it suffices to stick to local computation. Over any contractible open subset $U\subset M$, we fix the trivialisation of $T^\vee_M$ and $T_M$, and choose the compatible local coordinates $\{x^i, y^i\}_{i\in I}$, $\{x^i, z_i\}_{i\in I}$ respectively. The bracket structure $\{-,-\}_0$ is inherited from the Schouten bracket on $PV$, which locally is given by $\sum_i  \left(\overleftarrow{\frac\partial{\partial y^i}}\cdot \overrightarrow{\frac\partial{\partial z_i}} - \overleftarrow{\frac\partial{\partial z_i}}\cdot \overrightarrow{\frac\partial{\partial y^i}}\right)$. The complex \[ \big(\Omega_M\otimes \mathscr J\otimes PV, dx^i(\frac\partial{\partial x^i}-\frac\partial{\partial y^i})+  \sum_i \Pi^{ij}(x,y) z_i\frac\partial{\partial y^j}\big)\] is equipped with a bigrading by the form degree on $M$ and by the exterior power in the polyvector fields.

There exists a spectral sequence with the $E_1$ page given by 
\[E_1^{p,q}:= H^{p}\big( \Omega^p_M\otimes \mathscr J\otimes PV^q, dx^i(\frac\partial{\partial x^i}-\frac\partial{\partial y^i})\big),\] which computes Poisson cohomology of $M$ at $E_2$ page. On the other hand, the spectral sequence abuts to the total cohomology. 
In case of the symplectic manifold, the Poisson cohomology is isomorphic to $H^*(\Omega(M), d)$.

\end{proof}

\subsection{On boundary conditions}
\label{classicalbvthy}

Let $\Sigma$ be the two-dimensional worldsheet. Given the target $l_\infty$ algebra as described in Prop.~\ref{targetalgebra}, the space of field for our model is $\mathcal E=\Omega_\Sigma\otimes_{\mathbb R}\mathfrak g[1]$, which we shall also denote by $\Omega_\Sigma\otimes_{\mathbb R} \mathfrak h[1] \oplus \Omega_\Sigma\otimes_{\mathbb R}\mathfrak h^\vee$, following the notation used in the proof of Prop.~\ref{targetalgebra}. The shifted space of field $\mathcal E[-1]$ obtains an $l_\infty$ algebra structure by scalar extension. There is a degree $-1$ symplectic pairing on  $\mathcal E$ (viewed as symplectic complete bornological vector spaces \cite{KM, CG}) given by 
$$ (\alpha\otimes a)\otimes (\beta\otimes b)\mapsto \int_\Sigma \alpha\wedge \beta  \langle a, b\rangle_{1},$$
where $\langle-,-\rangle_1$ is the natural $ \Omega_M$-linear pairing between $ \Omega_M\otimes \mathcal T_M$ and $ \Omega_M\otimes \mathcal T^\vee_M$. 
This is an invariant paring with respect to the $l_\infty$ structure in the following sense. 

\begin{definition}
Let $(K,d)$ be a cdga over a field of character zero. Suppose that $\mathfrak h$ is a (curved) $l_\infty$ algebra over a cdga $(K,d)$. A pairing of degree $s$ is an invariant pairing if for all $n>0$, the linear map 
\[ \mathfrak g^{\otimes n}\to (K,d): (v_0, v_1,\cdots, v_{n-1})\mapsto \langle v_0, l_{n-1}(v_1,\cdots, v_{n-1})\rangle \]
is a graded skew-symmetric map of chain complexes. 
\end{definition}

\begin{remark} If the base cdga is a field of character zero, as in most physical setting, the obstruction to the graded anti-symmetrisation of the above map would be strictly zero. However, in the dg case, the obstruction could well be a coboundary of $(K,d)$. \end{remark}

If $\partial \Sigma = \emptyset$, this pairing induces the antibracket $\{-,-\}_0$ on the functionals, which can be viewed as $ \mathcal O_{\mathcal E}\cong (C^*(\mathcal E[-1]), \delta^{BV})$, and the CE differential $\delta^{BV}$ is the Hamiltonian vector field defined by the action functional $S$. Hence the equation of motion $ d\phi +\sum_n \frac1{n!}l_n(\phi^{\times n})  = 0$ is the Maurer-Cartan equation. So the classical, boundary-less BV theory can be concluded by
$$\delta^{BV}\circ \delta^{BV} = 0,\quad \text{where\,}\delta^{BV} f= \{S,f\}_0, \forall f\in \mathcal O_{\mathcal E}.$$

When the worldsheet $\Sigma$ has nonempty boundaries, we need a version of BV theory with boundaries\cite{CMR, CW}. The symplectic pairing $\int_\Sigma\langle-,-\rangle_1$ defined above is not compatible with the de Rham operator on $\Sigma$, and hence fails to be symmetric on $\mathcal E[-1]$. 
For this reason, the classical action $S$ fails to be the Hamiltonian function for $\delta^{BV}$, and hence the classical master equation bas to be modified 
$$ \delta^{BV}\circ \delta^{BV} = 0,\quad \delta^{BV} f= \{S,f\}_0+{\rm something\, on\,boundary}, f\in \mathcal O_{\mathcal E}.$$

Indeed, for general topological theories of AKSZ type (i.e., the extended $l_1$ operation on $\mathcal E[-1]$ contains the de Rham operator $d$ on $\Sigma$) we have the following result.
\begin{proposition}
Let $\Sigma$ be a manifold with boundaries, and let $\mathfrak g$ be a curved $l_\infty$ algebra over $(K,d^K)$, with operations $\{l_n\}_{n\geq0}$ and a shifted symplectic pairing $\langle-,-\rangle_*$of degree ${\rm dim}(\Sigma)-1$.
Let $S$ be the classical topological BV theory on $\Sigma$ as constructed by
\[S[\phi] = \int_\Sigma \langle \phi, \frac12d\phi+\sum_{n\geq0}\frac1{n!} l_n(\phi^{\otimes n})\rangle_{*}, \quad\phi\in\Omega_\Sigma\otimes\mathfrak g[1].\] The followings are equivalent.
\begin{enumerate}
\item The pairing $\langle-,-\rangle_*$ is invariant over $\mathfrak g$.
\item The pairing $\int_\Sigma\langle-,-\rangle_*$ is invariant with respect to the extended $l_\infty$ operations $\{ 1\otimes l_0, d\otimes id+ id\otimes l_1, id\otimes l_n\}_{n\geq2}$ on $\Omega_\Sigma\otimes\mathfrak g$ up to a boundary term over $\partial\Sigma$.

\item The action functional $S$ is an element in $C^*(\Omega_\Sigma\otimes\mathfrak g)$, whose variation is generated by the Maurer-Cartan functional up to a boundary contribution.
\end{enumerate}
\end{proposition}
\begin{proof} By linearity, if $\langle l_n(-,\cdots,-),-\rangle_*$ is graded skew-symmetric over $\Omega_\Sigma\otimes K$, so is $\int_\Sigma \langle l_n(-,\cdots,-),-\rangle_*$ over $\mathbb R\otimes K$. The only possible obstruction to an invariant pairing comes from the de Rham operator $d$, as $\int_\Sigma\langle d(-), -\rangle_*\pm \langle-, d(-)\rangle_* = \int_{\partial\Sigma}\langle-, -\rangle_*$. 

Consider the homogeneous term $ \int_\Sigma \langle \phi,\frac1{(n+1)!} l_n(\phi^{\otimes n})\rangle_* $ in $S$. The term being graded skew-symmetric with respect to $\Omega_\Sigma\otimes\mathfrak g$ if and only if it is an element in 
\[\widehat{\rm Sym}_{\boxtimes}(\Omega'_\Sigma\otimes \mathfrak g^\vee[-1])\cong C^*(\Omega_\Sigma\otimes\mathfrak g).\]
Due to the symmetry and the non-degeneracy on $\langle-,-\rangle_*$, the variation of $S$ leads to the Maurer-Cartan functional in the bulk of $\Sigma$. The only term that could possibly introduce a boundary integration involves de Rham operator $d$:
\[\delta \int_\Sigma \langle \phi, \frac{1}2 d\phi \rangle  = \int_\Sigma \langle \delta\phi, \frac{1}2 d\phi \rangle\pm \int_\Sigma \langle \phi, \frac{1}2 \delta d\phi \rangle.\]

Using the integration-by-part formula as given above, we have that \[\delta \int_\Sigma \langle \phi, \frac{1}2 d\phi \rangle  = \int_\Sigma \langle \delta\phi,  d\phi \rangle\pm \frac12 \int_{\partial\Sigma} \langle \phi,  \delta \phi \rangle,\]
hence
\[\delta S[\phi] = \int_\Sigma \langle \delta \phi , \sum_n \frac1{n!}l_n(\phi^{\otimes n})\rangle \pm \frac12 \int_{\partial\Sigma} \langle \phi,  \delta \phi \rangle. \]

On the other hand, suppose that $\delta S$ leads to a boundary contribution, then this has to come from the term with de Rham operator $d$. Now from the arbitrariness of $\delta \phi$ and $\phi$, we deduce that the skew self-adjointness of $d$ with respect to $\int_\Sigma$ is obstructed by a boundary term. 
\end{proof}

A simple fix of the problem is to consider a subspace of fields $\mathcal E^b\subset \mathcal E$ on which $\int_\Sigma\langle-,-\rangle_1$ is compatible with de Rham differential. More precisely, the subspace $\mathcal E^b$ needs to satisfies the following conditions.
\begin{itemize}
\item The subspace $\mathcal E^b[-1]$ inherits the $l_\infty$ structure. In particular, the de Rham differential $d$ must preserve $\mathcal E^b$.
\item The un-shifted pairing $\int_{\partial\Sigma}\langle-,-\rangle_1$ on $\mathcal E^b$ has to vanish, i.e., we shall impose certain boundary condition to a coisotropic subspace of $\mathcal E$. 
\end{itemize}

The first constraint can only be solved by Dirichlet boundary condition, so that is what we shall apply. The remaining constraints lead to choices in need of specifying the geometric structure of $M$ as well as a connection. While we would like the theory to produce interesting invariant structures for general Poisson manifolds, such choices are not desirable. So we are left with only one universal boundary condition: $\mathcal E^b\cong\Omega_\Sigma\otimes \mathfrak h[1]\oplus \Omega_{\Sigma,D}\otimes \mathfrak h^\vee$, with the natural pairing $\int_\Sigma \langle -,-\rangle_1$.

\section{BV quantization for theory with boundaries}
\label{quant}

\subsection{The propagators over configuration space}
\label{BV:prop}

For topological field theories, the propagator is a smooth form over the two-point configuration space, $\overline{{\rm Conf}(2,\Sigma)}$. This is also true when the boundary of $\Sigma$ is nonempty. As we shall see, the construction combines the results of Axelrod-Singer and Kontsevich \cite{AS94, Kont97}. The boundary condition is satisfied using a mirror charge method.

Let $\Sigma$ denote a Riemann surface with genus $g$ and $n\geq1$ boundary components. We shall first construct the diagonal class $\Xi$ as a smooth differential form over $\Sigma\times\Sigma$, which is used in defining the quantum BV laplacian and the propagator. We have that ${\rm dim} \, H^*(\Sigma) = 1+(2g+n-1)+0$, ${\rm dim} \, H^*(\Sigma,\partial\Sigma) = 0+(2g+n-1)+1 $ and ${\rm dim} \, H^*(\partial\Sigma) =2n$. It is easy to specify the representatives for those cohomology classes. In the following table we shall list those representatives, whose upper index denotes the form degree, and the lower index labels the base.
\begin{center}
\begin{tabular} {c|c|c}
$H^0(\Sigma) : \{\alpha^0\}$,&$H^1(\Sigma): \{\alpha^1_i\}_{i=1}^{n-1} \cup \{\gamma^1_{k}\}_{k=1}^{2g}$,& $H^2(\Sigma):\emptyset$ \\[5mm]
$H^0(\Sigma,\partial\Sigma) : \emptyset$,& $H^1(\Sigma,\partial\Sigma): \{d\beta^0_i\}_{i=1}^{n-1} \cup \{\gamma^1_{k}\}_{k=1}^{2g}$,& $H^2(\Sigma,\partial\Sigma):\{d\beta^1\}$ \\[5mm]
$H^0(\partial\Sigma) : \{\alpha^0\}\cup \{\beta^0_i\}_{i=1}^{n-1}$, &$H^1(\partial\Sigma):  \{\alpha^1_i\}_{i=1}^{n-1}\cup\{\beta^1\}$ &
\end{tabular}
\end{center}

In the above expression, we have already identified the classes under the map \[\cdots\to H^*(\Sigma,\partial\Sigma)\to H^*(\Sigma)\to H^{*}(\partial \Sigma)\to H^{*+1}(\Sigma,\partial\Sigma)\to\cdots.\]

So the diagonal form is entirely fixed by the Lefschtz-Poincare duality, which gives
\[\int_{\partial\Sigma} \alpha^0\wedge \beta^1=\int_{M} \alpha^0\wedge d\beta^1 = 1,\quad \int_{\partial\Sigma} \alpha_i^1\wedge \beta_j^0 =-\int_{M} \alpha_i^1\wedge d\beta_j^0= \delta_{ij} \]
\[\int_\Sigma \gamma^1_k\wedge \gamma^1_l= \omega_{k,l},\]
where $\omega_{k,l}$ is the sigh difference for the standard symplectic form on the first cohomology group of closed Riemann surfaces.

With the above data, we can construct a cocycle dual to the diagonal class, which can also be viewed as a projection operator projecting a form orthogonally onto the representatives of cohomology. Unlike the boundary-less case, we have two types of cohomology: namely $H^*(\Sigma,\partial\Sigma)$ and $H^*(\Sigma)$, so the class is an element in $\Omega_\Sigma\boxtimes\Omega_{\Sigma,\partial}$, which is given by
\[  \Xi := \omega_{k,l}\pi_1^* \gamma^1_k \otimes \pi_2^* \gamma^1_l+ \pi_1^* \alpha^0\otimes \pi_2^* d\beta^1+\pi_1^*\alpha_i^1\otimes \pi_2^* d\beta^0_i.\]

In this section and later we shall use extensively the configuration spaces of points in the worldsheet under Fulton-MacPherson compactification \cite{FMP}. Here our notation is as follows.
\begin{itemize}
\item The naive configuration space of $n$ points in $\Sigma$ is given by the Cartesian product $\Sigma^{\times n}$.
\item The configuration space of $n$ distinct points is denoted by ${\rm Conf}(n, \Sigma)$, which is $\Sigma^{\times n}$ deleting all the diagonal ideals
\item The FMP compactified configuration space is denotes by $\overline{{\rm Conf}(n, \Sigma)}$.
\item The boundary and corner strata of the configuration space are caused by two reasons: some points collides, or some hit the boundary. Suppose the set(s) $S$ contains all the points that collides to a single point, and $M$ contains the points that hit the boundary. We shall use $\partial_{S,M}\overline{{\rm Conf}(n, \Sigma)}$ to denote the corresponding strata.
\item When it is not necessary to specify the actual elements in $S$ and $M$, and when $S\cap M=\emptyset$, we shall just mention their cardinality.
\end{itemize}

\begin{theorem}
Suppose that $\Sigma$ is a compact Riemann surface with boundaries. Fix an embedding of $\Sigma$ into a finite dimensional Euclidean space with $\partial \Sigma$ as complete geodesic. The de Rham differential $d$ over the smooth forms $\Omega(\Sigma)$ (together with the obvious pairing using wedge product) has a parametrix $P^{an}\in \Omega^{1}(\overline{{\rm Conf}(2, \Sigma)})$ upon choosing the connection of $T_\Sigma$, which satisfies the following properties
\begin{enumerate}
\item $P^{an}|_{\partial_{2,0} \overline{{\rm Conf}(2, \Sigma)}}= \omega_\theta+bd^*  \alpha$ where $\omega_\theta$ is a fiber-wise volume form on the sphere bundle $Sph(T_\Sigma)$, and $\alpha$ is an $1$-form on the naive configuration space $\Sigma\times \Sigma$. 
\item The form $dP^{an}$ descends along the blow-down map, and $dP^{an} =  \Xi$ as $2$-form on $\Sigma\times \Sigma$, where $ \Xi$ is a cocycle dual to the diagonal class in $\Sigma\times\Sigma$.
\item $P$ satisfies Dirichlet boundary condition on one copy of $\Sigma$ inside $\overline{{\rm Conf}(2, \Sigma)}$.
\end{enumerate}
\label{propthm}
\end{theorem}

The approach extends results of Kontsevich and Axelrod-Singer to the boundary case, and the key procedure is as follows. The two-point configuration space, under FMP compactification, has a boundary strata keeping tracking of the process of points collapsing, which can be identified with the sphere bundle associated to the tangent bundle $T_\Sigma$ over the diagonal ideal $\Sigma\to \Sigma\times\Sigma$. We shall take the standard volume form on each fiber, and {\it glue} them over the whole sphere bundle $Sph(T_\Sigma)$, and then extend to the bulk of the configuration space. The gluing procedure is done by specifying the connection on $Sph(T_\Sigma)$ and take the horizontal form (with respect to the underlying principal bundle). Note that alternatively there are constructions with the requirement that $\Sigma$ being a parallelizable manifold, and then a trivialisation of $T_\Sigma$ would be assumed. For closed 2 manifolds, this is not possible due to the presence of $c_1$. For boundary case, the obstruction vanishes. However, to deal with the boundary, we shall need the double of the manifold $D\Sigma$, which is again closed. So in the following, we shall not assume the parallelizable structure on $\Sigma$, but use the connection explicitly similar to the procedure in \cite{BC98}. 

Firstly, consider the frame bundle over $\Sigma$ given by $\coprod_{x\in \Sigma} {\rm Iso}(\mathbb R^2, T_x\Sigma)$. For each $x\in \Sigma$, the isomorphism set has a free and transitive $GL(2)$ action from right, and is equivalent to $GL(2)$ set-theoretically. We now have a torsor structure on $Fr_\Sigma\equiv\coprod_{x\in \Sigma} {\rm Iso}(\mathbb R^2, T_x\Sigma)$. And with further refinements, this is identified with the frame bundle. The sphere bundle, considered as a boundary strata of configuration, is given by the associate bundle to the orthogonal frames $OFr_\Sigma\times_{SO(2)}S^1$. 

To construct a one-form on the sphere bundle form which restricts to volume form fiberwise, one consider firstly a $SO(2)$-invariant form on $OFr_\Sigma\times S^1$, which descends to the base along $OFr_\Sigma\times S^1\to OFr_\Sigma\times_{SO(2)} S^1$. 

The standard volume form on $S^1\subset \mathbb R^2$ is given by $\omega=\frac{xdy-ydx}{2\pi}$ (so that $\int_{x^2+y^2=1} \frac{xdy-ydx}{2\pi} = 1$), and the $so(2)$ action corresponds to a vector field $\xi=x\frac{\partial}{\partial y}-y\frac{\partial}{\partial x}$, which corresponds to the off-diagonal unitary matrix $X=\left(\begin{array}{cc}0 & 1 \\-1 & 0\end{array}\right)$. Now \[\mathcal L_{\xi}\omega = d\iota_\xi \omega = d(\frac{x^2+y^2}{2\pi} )=0,\] and \[\iota_\xi \omega = \frac1{2\pi}.\]  The vector field $\xi$ preserves the volume form. The connection for $OFr_\Sigma$ is a $so(2)$-valued one-form $\theta\otimes X$ on $OFr_\Sigma\times_{SO(2)} S^1$, and the induced connection on the sphere bundle is then given by the pullback of $\theta \otimes \xi$ locally via some sections in $\Gamma(\Sigma,OFr_\Sigma\times_{SO(2)} S^1)$. Since $so(2)$ is one-dimensional, the connection $\theta \otimes X$ satisfies \[\mathcal L_{\xi} \theta = 0\] which coms from the $SO(2)$-equivariance, and \[\iota_\xi \theta = 1\] which is due to the identity on vertical vectors. 

The form $\omega-\frac\theta{2\pi}$ can be shown to be a basic form on $OFr_\Sigma\times S^1$ by previous calculation, and therefore descends to $OFr_\Sigma\times_{SO(2)} S^1$. The fiber integration can be done on $OFr_\Sigma\times S^1$, so this gives $\int_{S^1} \omega-\frac\theta{2\pi} = \int_{S^1}\omega=1$. Further more, the form is invariant with respect to the antipodes action of the circle. Finally, to extend the form to the neighbourhood of zero section of $T_\Sigma$, and we denote the resulting form by \[\omega_\theta\equiv \frac1{2\pi}\frac{xdy-ydx}{x^2+y^2}-\frac\theta{2\pi}.\] 

\begin{proof} 

Choose a metric on $\Sigma$ such that the boundary $\partial \Sigma$ is complete geodesic. Now consider the local diffeomorphism
\[f: Nbd_\epsilon(\Delta)\backslash\Delta\subset \overline{{\rm Conf}(2, \Sigma)}\to : Nbd_\epsilon (T_\Sigma)\backslash0\cong Sph(T_\Sigma)\times (0,\epsilon)\]
given by
\[(z,w)\mapsto (z, \hat{\xi}(w,z), s(w,z)),\]
where $\hat{\xi}(w,z)$ is the unit tangent vector over $z$ along which there exists a unique geodesic from $z$ to $w$, and $s(w,z)$ is the length of the geodesic.
Given the covariant sphere volume form over $Nbd_\epsilon (T_\Sigma)\backslash0$ and a smooth function $\rho\in C^\infty_{(0,\epsilon)}$ such that $\rho \equiv 1$ over $(0,\epsilon/3)$ and $supp(\rho)\subset (0,2\epsilon/3)$, define 
\[P^{an}(z,w):=f^* (\omega_\theta \cdot \pi^*_\epsilon \rho)(z,w) ,\] 
where $\pi_\epsilon: Sph(T_\Sigma)\times (0,\epsilon)\to (0,\epsilon)$ is the natural projection. Note that $P^{an}$ can be extended to the diagonal strata without any difficulty, which under the local diffeomorphism $f$, corresponds to the thickened sphere bundle $Sph(T_\Sigma)\times [0,\epsilon)$. From here on, we shall mention $P^{an}$ as the extended form.

On $T_\Sigma$, $d\omega_\theta$ is well-defined along the zero-section. Moreover, $\pi^*_\epsilon d\rho$ vanishes near the zero section inside $Nbd_\epsilon (T_\Sigma)\backslash0$. So \[dP^{an} = f^* (d\omega_\theta \cdot \pi^*_\epsilon \rho-\omega_\theta \cdot \pi^*_\epsilon d\rho)\] is a well-defined 2-form on the naive configuration space $\Sigma^{\times2}$. In the following we shall not distinguish explicitly $dP^{an}$ with its descended version.

Finally we need to solve the issue of boundary condition on one of the variables in $P$. To achieve this, we shall use the mirror charge method. The double $D\Sigma\equiv \Sigma\cup_{\partial\Sigma}\Sigma^{op}$ of the manifold $\Sigma$ can be defined whenever the complete geodesic boundaries $\partial\Sigma$ are given, with an involution $\tau$ interchanging points from different sides of the boundary. 
\[\tau: D\Sigma \to D\Sigma: z\mapsto z^{op} ,\]
where $\Sigma^{op}$ is the manifold equipped with opposite orientation, and $z^{op}$ is the the mirror image of $z\in \Sigma$ (resp., $z\in \Sigma^{op}$) in $\Sigma^{op}$ (resp. $\Sigma$). Further more, the embedding of $\Sigma$ entirely determines that of $D\Sigma$. In particular, the connection of $Fr_\Sigma$ determines uniquely a connection on $Fr_{D\Sigma}$.

The reflective symmetry induces one at configuration space level 
\[id\times \tau: \overline{{\rm Conf}(\{z, w\},\Sigma)}\to \overline{{\rm Conf}(\{z,\tau w\},\Sigma,\Sigma^{op})}\subset \overline{{\rm Conf}(2,D\Sigma)}.\]  
In the configuration $\overline{{\rm Conf}(2,D\Sigma)}$, the local structure near the diagonal is similar. Let $f^* (\omega_\theta \cdot \pi^*_\epsilon \rho)(z,w) $ denotes the propagator construction similarly as above over $D\Sigma$, and we shall restrict the form to $\overline{{\rm Conf}(\{z,\tau w\},\Sigma,\Sigma^{op})}$ and pullback via $id\times \tau$. Now the pull-ed back form is denoted by $f^* (\omega_\theta \cdot \pi^*_\epsilon \rho)(z,\tau w) $.
Now we can define the propagator \[P^{an}(z,w)=f^* (\omega_\theta \cdot \pi^*_\epsilon \rho)(z, w)-f^* (\omega_\theta \cdot \pi^*_\epsilon \rho)(z,\tau w)\] with Dirichlet boundary condition on the variable $w$. The second part of the one-form is exact near the diagonal boundary strata, and vanishes when restricted to the fiber boundary strata.

One checks that on $\Sigma^{\times2}$ the form $d(P^{an}) \sim  \Xi$ cohomologically. Firstly, $d(P^{an})$ is an $n$-dimensional locally exact form over $\Sigma\times\Sigma$, so it is a close form on $\Sigma\times\Sigma$. The latter being a manifold with corners, so one could use a version of Lefschtz-Poincare duality on $\Sigma\times\Sigma$ to determine the cohomological class of $d(P^{an})$. The duality pattern is entirely defined by that on $\Sigma$. Consider the cohomological representatives of $H(\Sigma)\otimes H(\Sigma, \partial\Sigma)$, $d(P^{an})$ has to be a linear combination of basis $\{ \pi_1^* \gamma^1_k \otimes \pi_2^* \gamma^1_l, \pi_1^* \alpha^0\otimes \pi_2^* d\beta^1,\pi_1^*\alpha_i^1\otimes \pi_2^* d\beta^0_i\}$. We fix the coefficients by considering the bilinear pairing on $H(\Sigma)\otimes H(\Sigma, \partial\Sigma)$:
\begin{eqnarray*}&&\int_{\overline{{\rm Conf}(2, \Sigma)}} dP^{an}(z,w) \wedge \pi_1^* f \wedge \pi_2^* g=  \int_{\Sigma\times\Sigma}dP^{an}(z,w) \wedge \pi_1^* f \wedge \pi_2^* g
 \\&=& 
\int _{\partial_{\{2\},\emptyset} \overline{{\rm Conf}(2, \Sigma)}} (P^{an}(z,w) \wedge \pi_1^* f \wedge \pi_2^* g)|_\partial+ \int _{\Sigma\times\partial \Sigma} (P^{an}(z,w) \wedge \pi_1^* f \wedge \pi_2^* g)|_\partial\\&&+\int _{\partial\Sigma\times\Sigma} (P^{an}(z,w) \wedge \pi_1^* f \wedge \pi_2^* g)|_\partial\\
&=& \int_\Sigma f\wedge g,
\end{eqnarray*}
where $f\in \Omega^{cl}(\Sigma,\partial\Sigma)$ and $g\in \Omega^{cl}(\Sigma)$.
Similar result holds when we take arbitrary closed form in $\Omega(\Sigma\times\Sigma)$ which satisfies the Dirichlet condition on the first copy of $\Sigma$. This means that $dP^{an}$ is in the same cohomology class as the cohomological dual of the diagonal cycle in $\Sigma\times\Sigma/\Sigma\times\partial\Sigma$.
Since $dP^{an}$ and $\Xi$ are in the same cohomology class of $\Sigma\times\Sigma$, we can simply use an exact form $d\alpha$ to make up for their discrepancy.

Our last possible issue is, whether the newly added correction $\alpha$ respects the same desired boundary condition. Note that in the calculation of $dP$, $d [f^* (\omega_\theta \cdot \pi^*_\epsilon \rho)(z, w)-f^* (\omega_\theta \cdot \pi^*_\epsilon \rho)(z,\tau w)]$ respects the same boundary condition, and hence when choosing the form $\alpha$, it must also respect the same boundary condition.
\end{proof}

With the theorem we have constructed the analytic part $P^{an}$ of the propagator over $\overline{{\rm Conf}(2, \Sigma)}$, which is real valued. Since in our Poisson model (and in fact in all TFTs) the propagator factories as \[P(x,y) = P^{an}(x,y)\otimes P^{alg},\] where $P^{alg}$ is the algebraic part given by the target $l_\infty$ algebra $\mathfrak g$. By the invariant property of the shifted symplectic pairing, $P^{alg}$ satisfies the relation \[(l_1\otimes id + id\otimes l_1)P^{alg} = 0\] automatically. (Another reason is that, $P^{alg}$ induces the identity operator on $\mathfrak h[1]$ via the shifted symplectic pairing, and the identity operator commutes with $l_1$.)
In the Poisson case, choosing a $\Omega_M$-linear basis $\{\mathbb E_X^i\}$ on $\mathfrak h$ (and thereby the basis on the dual space), we have that
\[P^{alg}=\sum_i \mathbb E_X^i[1]\otimes \mathbb E^\eta_i\in \mathfrak h[1]\otimes \mathfrak h^\vee.\] 

\begin{remark}
A propagator $P$ gives rise to a homotopy $s_P$ on the space of field, $\mathcal E^b$ between the identity operator and the operator $\iota\circ \Xi$, which can be further adjusted to satisfy the following conditions \cite{CM}
\begin{itemize}
\item $s_P\circ \iota = 0 = \Xi\circ s_P$,
\item $s_P\circ s_P = 0$.
\end{itemize}
To see this, firstly note that $s_P(\phi)(y):=\int_{b_2^{-1}(y)}  \phi(x)P(x,y)$ whenever $\phi\in \Omega_{\Sigma,D}\otimes \mathfrak h^\vee$ and $s_P(\phi)(y):=\int_{b_1^{-1}(y)} P(y,x) \phi(x)$ whenever $\phi\in \Omega_\Sigma\otimes\mathfrak h $. The map $b_i$ is the forgetting map \[b_i: \overline{{\rm Conf}(\{x_1,x_2\},\Sigma)}\to \overline{{\rm Conf}(\{x_i\},\Sigma)} = \Sigma, \forall i=1,2. \] So the boundary conditions are preserved by $s_P$.
To achieve the above extra properties, we do a series of replacement $s_P\mapsto (id-d\circ s_P-s_P\circ d)\circ s_P\circ (id-d\circ s_P-s_P\circ d)$ to satisfy the first property, as well as $s_P\mapsto s_P\circ d\circ s_P$ to satisfy the second property assuming the first. Obviously those operations do not change the boundary property of $P$, nor the homotopy type.
\end{remark}

Finally, it is important to note that our construction of propagators depends on a partial choice of metric on $\Sigma$ (so as to make the boundary complete geodesic), a connection $\theta$ on $T_\Sigma$, as well as a cutoff function $\rho$. In the following we shall see that the theory exhibits many interesting observables which are independent of the latter two choices.

\subsection{Some vanishing result}

In this section, we shall prove some vanishing results in the given theory. 

We recall the Feynman graph computation firstly. The propagator $P(z,w)$ is a $\mathfrak g\otimes \mathfrak g$-valued smooth one-form over $\overline{{\rm Conf}(2,\Sigma)}$ subject to the chosen boundary condition. By a vertex we mean a $C^*(\mathfrak g)$-valued distribution over $\Sigma$. The contraction between vertices(, labeled by all the vertices $V(\Gamma)$ of a given graph $\Gamma$), and propagators(, labeled by all inner edges $E(\Gamma)$), is defined by the contraction between $\mathfrak g$ and $\mathfrak g^\vee$, and the integration against the configuration space of the vertices. 

Suppose a Feynman graph $\Gamma$ has internal vertices (i.e., vertices coming from terms in the interaction $I$) labelled by $V(\Gamma)$ and external vertices (i.e., whose vertices do not just come from the classical action, see Sec.~\ref{secobserv} for examples) labelled by $W(\Gamma)$, then the integral is the fiber integration along the forgetful map \[ \overline{{\rm Conf}(V(\Gamma)\cup W(\Gamma), \Sigma)}\to  \overline{{\rm Conf}(W(\Gamma), \Sigma)}.\] Suppose that $\Gamma$ has inner edges (i.e., propagators) labeled by $E(\Gamma)$ and external edges decorated by field $\phi$, the analytic part of Feynman integration is written as
\[\int _{\overline{{\rm Conf}(V(\Gamma), \Sigma)}} \wedge _{e\in E(\Gamma)} P^{an}(e_{in}, e_{out}) \wedge_{v\in V(\Gamma)} \phi^{an}(v).\]The computation for the algebraic sector is given by the natural pairing between $\mathfrak g$ and $\mathfrak g^\vee$.

The configuration space approach in defining the propagator fully determines the evaluation of all Feynman graphs, except for those with tadpoles. By a tadpole we mean a graph with an edge that starts and ends with the same internal vertex, and hence receives the contribution from $P^{an}(z,z)$. In quantization approaches where the diagonal part $P^{an}(z,z)$ is well-defined, the tadpole graphs could lead to an anomaly in BV quantization, for example see \cite{LL16}. In our approach, the analytic propagator as a one-form on the two-point configuration space, which does not descend along the pushdown map, so $P^{an}(z,z)$ is undefined. 

There are essentially two ways to solve the issue of tadpole. If the Euler class in $H^2(\Sigma, \partial\Sigma)$ is trivializable by a one-form in $\Omega^1(\Sigma, \partial\Sigma)$, then it is possible to define $P^{an}(z,z)$ by this one-form without spoiling anything. At differential level, the triviality of Euler class $\chi_{\Sigma,\partial\Sigma}$ provides a no-where vanishing vector field on $\Sigma^\circ$, allowing the point-separation technique as in \cite{CF00} to apply. 

If $\chi_{\Sigma,\partial\Sigma}\neq0$, one can not find a nonsingular expression for $P^{an}(z,z)$. Then one is forced to discuss only graphs without tadpoles \cite{CMR}, or equivalently, setting all tadpoles to zero. This seemingly naive operation has an explanation at the level of graph complex in \cite{CW}. Here we shall show that ignoring tadpoles can be done in a consistent way. By consistent, we mean all the allowed operations, such as multiplication, changing the integration order, acting by $Q$, restricting graphs to the boundary strata of configuration spaces, would not result in tadpoles. 

\begin{proposition}
The Feynman graphs involving tadpoles can be set to vanish consistently.
\label{tad}
\end{proposition} 

\begin{proof}

Tadpole can only happens very locally at each vertices, the singularity is along the diagonal ideal, so tadpole at different vertices can be multiplied without causing further trouble. 

For dimension reason, the theory can only have one or two tadpoles at most. Due to the lack of parallel edges in our theory (stated and proved below in Prop.~\ref{par}), the number of allowed tadpoles is reduced to one. So we shall denote such tadpole vertices by $(\hslash\partial_P I)$.

Now all the connected Feynman graphs can be enumerated\footnote{We will show that the formal expression involves only finite graphs at each finite order of $\hslash$ in the next section, Thm.~\ref{finitesumthm}.} by 
\[{\rm log}[{\rm exp}(\hslash \partial_P) {\rm exp}(\frac {I}\hslash)]={\rm log}[{\rm exp}(\hslash \partial_P') {\rm exp}(\frac {I + \hslash(\partial_P I)}\hslash)],\]
where $\partial_P'$ is the propagator contraction operator which vanishes when acting on a single vertex. The subset of graphs ${\rm log}[{\rm exp}(\hslash \partial_P') {\rm exp}(\frac {I}\hslash)]$ is well-defined.

The only potential danger is, when we have $Q$ acting on graphs. We have that $QP^{an}(z,z) = \Xi(z,z)$, which is well-defined and provides a cohomological representative for Euler class in the boundary case by the Lefschtz-Poincare duality. But it is easy to see that 
\[{\rm log}\Big([Q,{\rm exp}(\hslash \partial_P') ]{\rm exp}(\frac {I}\hslash)\Big)\] contains no tadpole. Thereby all potential problems is when $Q$ acts on $(\hslash\partial_P I)$, which is not of our concern.
\end{proof}

Due to this reason, in the rest of the paper we shall not discuss tadpole graphs. The operator $\partial_P$, and $[Q, \partial_P]$ shall be understood as $\partial_P'$, and $[Q, \partial_P']$ without further confusing.

By parallel edges we mean a graph that contains a bi-gon, and the orientations of both propagators are parallel.

\begin{proposition}
The graphs involving parallel edges vanish.
\label{par}
\end{proposition}
\begin{proof}
This is a pure algebraic result.
Consider the algebraic part of the graph
\[\langle\Pi_m(\mathbb E_{\eta}, \cdots), \mathbb E_{\eta}\rangle_1 \langle l_n(\mathbb E_X, \cdots),\mathbb E_X\rangle_1 
\]
where the pairing is the $1$-shifted symplectic pairing on the target $l_\infty$ algebra $\mathfrak h \oplus \mathfrak h^\vee[-1]$. The pairing is $\Omega_M$-linear, and hence the symmetric property depends only on the degree of the $\mathcal T_M[1]\oplus \mathcal T^\vee _M$-component. I.e., given $\alpha\otimes X$ and $\beta\otimes Y$ in $\Omega_M\otimes(\mathcal T_M\oplus \mathcal T^\vee _M[-1])$, \[l_n(\alpha\otimes X, \beta\otimes Y) = (-)^{|\beta|\cdot|X|}\alpha\wedge\beta\otimes l_n(X, Y, \cdots) 
.\] Since the whole expression is $\Omega_M$-linear, one only need to consider the evaluation on the basis $\{\mathbb E_X^i, \mathbb E^\eta_i\}_{i\in I}$ inside $\Omega_M^0\otimes(\mathcal T_M\oplus \mathcal T^\vee _M[-1])$-component. The $\mathcal T_M$-component has even degree, while the $\mathcal T^\vee _M[-1]$ component has odd degree. Since the algebraic propagator $\mathbb E_{\eta}\otimes\mathbb E_X$ provides a contraction, this would vanish the above expression.

\end{proof}

\begin{proposition}
The graphs involving the $l_0$-vertex vanish. 
\label{prop:l0}
\end{proposition}

\begin{proof}
The element $l_0$ is a degree $2$ element in $\Omega^0_\Sigma\otimes \mathfrak h$. The fact that $l_0$ is a constant 0-form on $\Sigma$ is due to the fact that the de Rham differential on $\Sigma$ squares to zero, and $l_0$ is obtained by the tensor product of $1\in \Omega^0_\Sigma$ and the curving in $\mathfrak h$. So in the vertex $\int_\Sigma\langle l_0, \eta\rangle_1$, the field $\eta$ contributes its 2-form component. But if so, any propagator connecting that $\eta$ to other fields must vanish, since the propagator is a 1-form by construction.
\end{proof}

\subsection{Quantum master equation}
\label{sec:qme}

In this section we shall look into analysis of non-vacuum Feynman graphs and show that the quantum master equation holds for effective action. As a topological theory, we have that
$$
I_{\rm eff} [\phi]= {\rm log}(e^{\hslash \partial_P} e^{\frac{I}{\hslash}})=\sum_\Gamma \frac1{|{\rm Aut}\,(\Gamma)|} \int_{\overline{{\rm Conf}(\Gamma,\Sigma)}}  \prod_{e\in {\rm E}(\Gamma)} \hslash \partial_{P^e}   \prod_{v\in V(\Gamma)} \frac{I_v}\hslash ,
$$
where the summation is over all connected Feynman graphs with only bulk vertices. Indeed, if one of the vertices hits the boundary, the integration over the configuration vanishes due to the boundary condition one chooses for the propagator and for the background field $\phi$. For each graph $\Gamma$, we use ${\rm E}(\Gamma)$ to denote the set of its internal edges (propagators), and likewise $V(\Gamma)$ for the set of vertices.

While the previous section about the vacuum graphs has the result largely depend on the general properties of TFT, the non-vacuum graphs depends more on the Lagrangian description of each independent theory. In our theory, the propagator corresponds to a two-point correlation 1-form on the configuration space of two points in $\Sigma$, which is graphically described by an inner edge with direction, which points from the $\eta$ field side to the $X$ field side.  Vertices can have arbitrary valency greater or equal\footnote{There is an $l_0$ operation in the target $l_\infty$ algebra $\mathfrak g$, so the classical action exhibits a valency one vertex. However, as we saw previously, such vertex amounts to vanishing graphs.} to $2$. 
From each graph we shall obtain a (not necessarily local) homogeneous functional over the external fields, whose homogeneity is given by the number of half edges in the corresponding graph. 
\begin{figure}
\begin{center}
\includegraphics[width=2in]{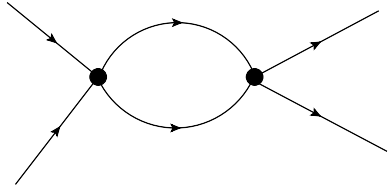}
\end{center}
\caption{\footnotesize An example of directed graph with two bulk vertices.}
\label{bulk1}
\end{figure}
\begin{example}
In Fig.~\ref{bulk1} we give an example of graphs that are under consideration, which involves only bulk vertices. This graph gives rise to a functional on the background field, whose homogeneous component is given by the map
\begin{eqnarray*} \mathcal E^{\otimes 4}\to \mathbb R&:& (\phi_1, \phi_2, \phi_3, \phi_4)\mapsto \\
&& ({\rm symmetry\,factor})\cdot \int_{\overline{{\rm Conf}(\{z,w\}, \Sigma)}} \partial^2_{P^{an}(z,w)} \partial_{\phi_1(z)}\partial_{\phi_2(z)} \partial_{\phi_3(w)}\partial_{\phi_4(w)}\cdot\\
&&\cdot\langle \phi(z), \frac1{4!} l'_3(\phi^{\otimes 3}(z))\rangle_1\langle \phi(w), \frac1{4!} l'_3(\phi^{\otimes 3}(w))\rangle_1,  \end{eqnarray*}
where $\phi\in \mathcal E$. 
The expression factorises into the analytic part on the configuration space, and the algebraic part which encodes the $l_\infty$ algebra. 
\[ \int _{\overline{{\rm Conf}(2, \Sigma)}} \omega_\Gamma  =\int _{\overline{{\rm Conf}(2, \Sigma)}}  \Phi\left(X_1(z), X_2(z) ,\eta_3(w) ,\eta_4(w)\right) P^{an}(z,w) P^{an}(z,w).\]
Interpreted as a quartic functional over $\mathcal E$, the above integration needs to make sense for any external field, i.e., for any choice of $X_1$, $X_2$, $\eta_3$ and $\eta_4$.
\label{ex1}
\end{example}

To make the sum mathematically sensible, it is important to show that there are only finitely many summands for each term in $I_{\rm eff} [\phi]$. The terms in $I_{\rm eff} [\phi]$ is labeled by the power of $\hslash$ and by the homogeneous degree viewed as a functional over $\Omega_\Sigma\otimes \mathfrak g[1]$. In terms of the graphs, one needs to show that for a given number of external legs, there are only finitely many Feynman graphs. This is trivially true if the classical theory contains no bivalent vertices. However, both the $l_1$ operation on $\mathfrak h$ and the bivector field $\Pi$ leads to such terms. So we shall give a proof for the statement.

\begin{theorem}
For each homogeneous term in $I_{\rm eff} [\phi]$ up to a given number of loops (, or genes, or powers in $\hslash$), there are only finitely many connected Feynman graphs in the summation.
\label{finitesumthm}
\end{theorem} 

\begin{proof}
Consider one graph $\Gamma$ in Poisson theory with $n$ external legs and $l$ loops. Suppose the number of bivalent vertices in $\gamma$ is given by $v_2$
and the number of at least trivalent vertices is given by $v$. Let $e$ denotes the number of propagators. So we have that 
\[ v+v_2-e+l = 1.\] 
On the other hand we can count the number of inner edges from the vertices side:
\[n+2e\geq 3v+2v_2.\]
Combing the two relations, one obtains that 
\[v\leq 2l+n-2,\]
so there are at most $2l+n-2$ as many at least trivalent vertices. Together with the bound for external legs, there are only finitely many choices of at least trivalent vertices that show up in those graphs. It only remains to show that $v_2$ is bounded. If so, there are only finite many choices for bivalent vertices that can show up in those graphs. So the total number of graphs is finite.

To show the bound of $v_2$, consider any connected graph $\Gamma$ and delete all the at least trivalent vertices. From the topology of the graph, we must be left with finite number of linear graphs which involves only bivalent vertices. The proposition can be shown by noting that there can not exist a linear graph with infinitely many bivalent vertices. (The smallest linear graph is the propagator without vertices.) Within a linear graph, the number of $\int_\Sigma\langle \eta, l_1(X)\rangle_1$ vertex can not exceed ${\rm dim} M$, since $l_1(X)$ is in the nilpotent ideal $\Omega^{\geq1}_M\otimes\mathfrak h$. The number of $\int_\Sigma\langle \eta, \Pi_1(\eta)\rangle_1$ vertex can not be larger than $1$, otherwise those two vertices need to be connected by a bivalent vertex quadratic in $X$ somewhere, which is not possible. This completes the proof.
\end{proof}

To make sense of the Feynman graph calculation, it still remains to show that the integration associated to each valid graphs is finite. 
But since the FMP configuration space is compact, that is automatic. 

The main result in this section is Thm.~\ref{thmqme}, stating the QME and a parameterised version, the proof of which we shall now present.

\begin{proof}

Consider $QI_{\rm eff}[\phi]:=d_\Sigma I_{\rm eff}[\phi] = I_{\rm eff}[\phi; d\phi]$ \footnote{Note that here the $l_1$ operation was taken into account in the interaction $I$ of the action, and correspondingly the operator $Q$ showing up in classical Master equation and quantum Master equation contains only the de Rham operator over $\Sigma$, as oppose to $QI_{\rm eff}[\phi]:=d I_{\rm eff}[\phi] + l_1 I_{\rm eff}[\phi]$ in \cite{Co11, GG_ahat}. Luckily here, with the insertion of the bivalent vertex would not result in the issue of infinite graphs, as we saw in Thm.~\ref{finitesumthm}.}.  If $I_{\rm eff}^{(k)}$ is a functional of homogeneous degree $k$, then $I_{\rm eff}^{(k)}[\phi; d\phi]$ is a short notation of $k$-multilinear functional
\[\sum_{i=1}^k I_{\rm eff}[\phi_1, \cdots, d\phi_i, \cdots, \phi_k], \forall \phi_1, \cdots\phi_k\in \mathcal E_\Sigma.\]

Let $\phi\in \mathcal E_\Sigma$, by Stokes theorem, we have
\begin{eqnarray*}QI_{\rm eff}[\phi] &=&\sum_\Gamma \int_{\partial \overline{{\rm Conf}(\Gamma,\Sigma)}} ({\rm boundary\,contribution}) \,\\
&&-\int_{\overline{{\rm Conf}(\Gamma,\Sigma)}} ({\rm one \,propagator \,replaced\, by \,}\Xi) ,
\end{eqnarray*}
which, in graph expansion, is given by
\begin{eqnarray*} QI_{\rm eff}[\phi] &=& - \sum_{\Gamma} \frac1{|{\rm Aut}\,(\Gamma)|} \sum_{g\in {\rm E}(\Gamma)} \int_{\overline{{\rm Conf}(\Gamma,\Sigma)}} \hslash \partial_{\Xi}^g \cdot\prod_{e\in {\rm E}(\Gamma)\backslash \{g\}} (\hslash \partial^e_{P}) \cdot  \prod_{v\in V(\Gamma)} \frac{I_v[\phi]}\hslash \\
 &&+\sum_{\Gamma} \frac1{|{\rm Aut}\,(\Gamma)|} \int_{\partial_\delta \overline{{\rm Conf}(\Gamma,\Sigma)}} \prod_{e\in {\rm E}(\Gamma)} \hslash \partial^e_{P} \cdot \prod_{v\in V(\Gamma)} \frac{I_v[\phi]}\hslash
\end{eqnarray*}

The first line decouples into two classes according to the topology of subgraph of $\Gamma$ after deleting one inner edge $g$. I.e., the edge $g$ either attaches to a single {\it connected} component of $\Gamma$, or connects two {\it disjoint} subgraphs in $\Gamma$, as shown in Fig.~\ref{1PI}. As is clear, the first class corresponds to the term $\hslash \Delta I_{\rm eff}[\phi]$, while the second class corresponds to $\{I_{\rm eff}[\phi], I_{\rm eff}[\phi]\}$. \begin{figure}
\begin{center}
\includegraphics[width=4in]{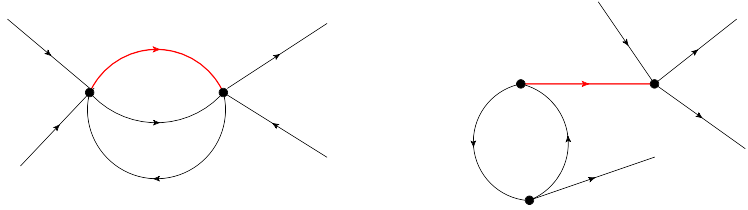}
\end{center}
\caption{\footnotesize An example of Feynman graph where the edge $g$, shown in red colour, attaches to a single connected component as shown in left, and a Feynman graph where the edge $g$ connects two disjoint subgraphs as shown in the right.}
\label{1PI}
\end{figure}

The second line involves integrations over the boundary strata of the configuration space, which could come from three cases: 
\begin{enumerate}
\item one or more vertices in $\Gamma$ hits the boundary $\partial\Sigma$,  
\item two vertices collapse into a single point in the bulk $\Sigma^\circ$,
\item three or more vertices collapse into a single point in the bulk $\Sigma^\circ$.
\end{enumerate}

For the first case, due to the boundary condition on propagators and on external fields, there can not be nontrivial contribution. 

For the third case, as we shall see, vanishing results following Kontsevich's arguments apply. 
Since the integration is supported on some diagonal, it is of form $Sph(T^{\Gamma'}_\Sigma)\times_\Sigma {\rm Conf}(\Gamma/\Gamma', \Sigma)$, where $\Gamma'$ is the collapsing subgraph of $\Gamma$, and $\Gamma/\Gamma'$ is the reminiscent of $\Gamma$ after the collision. The fiberproduct is defined via the base map of the sphere bundle, and the forgetful map of those non-collapsing vertices. Overall the integration is nontrivial only when there are $2n-3$ propagators where $n$ is the number of collapsing vertices. The external lines are all supported on ${\rm Conf}(\Gamma/\Gamma', \Sigma)$, and hence do not contribute to the fiber integration. Now shall focus on the collapsing sub-graph $\Gamma'$, which has $n$ vertices and $2n-3$ inner edges. So there must exists at least one vertex in $\Gamma'$ with less than one attaching edges. 

Those $2n-3$ inner edges correspond to the restriction of propagators onto the boundary strata, and for each propagator, only the covariant sphere volume form part survives, i.e., $P_{x,y}|_{x\to y} = \omega_{\theta}|_{x,y}$. The fiber integration corresponds to an integration along the $2n-3$-dimensional sphere in
\[\underbrace{OFr \times_\Sigma OFr\times_\Sigma \cdots OFr}_{(n-1)\text{-fold fiber-product}}\times S^{2n-3},\]
 (which shall be denoted by $OFr^{n-1}\times S^{2n-3}$) by consider the following commutative graph
\[\xymatrix{
OFr^{n-1}\times S^{2n-3} \ar[r]^{\bar p} \ar[d]^{\bar \pi} & OFr^{n-1}\times_{SO(n)^{n-1}} S^{2n-3} \ar[d]^\pi\\
OFr^{n-1} \ar[r]^p & \Sigma
}.\] 

When $n>2$, the integrand $ \omega_{\Gamma'}$ over $OFr^{n-1}\times_{SO(n)^{n-1}} S^{2n-3}$ satisfies that $p^*\pi_*\omega_{\Gamma'} = \bar\pi_*\bar p^*\omega_{\Gamma'} $ due to the commutativity. The expression $\bar p^*\omega_{\Gamma'} $ is a polynomial over $\theta$ valued in $\Omega^*_{S^{2n-3}}$. Integration along $\bar\pi$ leads to a basic form depends polynomially over $\theta$. So $\pi_*\omega_{\Gamma'}$ is a characteristic form. But there is no such form so $\pi_*\omega_{\Gamma'}$ is independent of $\theta$. So we only need to consider the boundary integration of differential forms depend polynomially over the pullbacks of standard $S^1$-volume form $\omega$, and this is exactly the setting where Kontsevich's vanishing result (c.f., 6.6 in \cite{Kont97}) applies. 

For the second case where $n=2$, a similar analysis shows that one only need to consider the case where the integrand (i.e., the shrinking edge) is independent of $\theta$. Then the fiber integration $\int_{S^1} \omega$ leads to a unit $1$. Effectively, the shrinking edge induces the merging of its end points, which can be classified into three sub-cases. 
\begin{enumerate}[i)]
\item The shrinking edge connects one $l_n$-vertex to another $l_m$-vertex.
\item The shrinking edge connects one $l_n$-vertex to an $\Pi$-vertex.
\item The shrinking edge connects an $\Pi$-vertex to another $\Pi$-vertex.
\end{enumerate}
In each case, the total results vanish due to the curved $l_\infty$ identity, the graded (anti)commutativity of $l_n$ and Poisson bi-vector field $\Pi$, and the Jacobian of the latter respectively\footnote{Recall that for the identity of curved $l_\infty$ algebra to hold, we need to include the $l_1$-vertices in the interaction term $I$.}.
This completes the proof for the statement that $S_{\rm eff}$ satisfies QME.

Next we are going to show that the dependence on the gauge choices we made in defining the propagator do not change the result. The gauge choices of propagator involve the following aspects: 
\begin{enumerate}
\item a choice for the representative for the cohomology of $\Sigma$ via the diagonal form $\Xi$, 
\item a choice of connection on $T_\Sigma$, 
\item a choice for the cutoff function $\rho$, 
\item and a choice of metrics of $\Sigma$ near the boundaries, such that the boundaries are complete geodesic sub-manifolds. 
\end{enumerate}
Consider any one-parameter family of the above choices, i.e., a smooth map from an interval $\mathbb I$ to the infinite dimensional moduli space $\mathcal M$ of gauge-choices. We can always realise such a family into a family version of propagator defined on $\Sigma\times \mathbb{I}$. 

Firstly, the choice of $\Pi_h$ can be encoded in the parameterised construction in the following way. When we extend the representatives of $H^*(\Sigma)$ and $H^*(\Sigma, \partial\Sigma)$ by considering the $t$-family of forms $\alpha(t)\in \Omega^{cl}(\Sigma)$, the difference $\alpha(t_1)-\alpha(t_2)$ is an exact form on $\Sigma$ for any $t_1, t_2\in \mathbb{I}$. Now $\frac{\partial}{\partial t}\alpha(t) $ is an exact form on $\Sigma$ parameterised also by $t$. Note that this condition also says that the representative for $H^0$ classes can not be deformed, and so has to be independent of $t$. For the non-trivial case, the above information can be encoded into a close form  $\tilde \alpha:=\alpha(t)+\alpha_h(t)dt$ on $\Sigma\times \mathbb{I}$, which represents a given cohomology class of $H^*(\Sigma\times \mathbb{I})\cong H^*(\Sigma)$. The closeness of $\tilde\alpha$ specifies precisely that $d_\Sigma \alpha(t) = 0$ and $d_\Sigma \alpha_h(t)+(-)^{|\alpha|}\frac\partial{\partial t}\alpha(t) = 0$.
The above argument does not depend on the choices of boundary conditions, and hence also works for representatives of $H^*(\Sigma, \partial \Sigma)$. So overall, we are able to promote $\Xi$ to a $t$-parameterised version to keep track of choice of cohomological representatives, which we denote by $\tilde {\Xi}$, such that
$\tilde \Xi(x,y;t) = \Xi(x,y;t)+\Xi_h(x,y;t)dt$. By construction, $\Xi$ is $d_\Sigma$-closed\footnote{Here and later within this proof we use $d_\Sigma$ to denote the pullback of de Rham operator over various configuration spaces of $\Sigma$.}, while $\Xi_h$ satisfies that \[d_\Sigma \Xi_hdt+d_t\Xi = 0.\] 

Secondly, different choices of connection $\theta$ on $OFr_\Sigma$ and the cutoff function $\rho$ can be connected by promoting both to the parameterised version. Namely consider a connection $\tilde \theta$ on $OFr_\Sigma\times \mathbb{I}$ (, the latter understood as a pullback bundle of $OFr_\Sigma$ along $\Sigma\times\mathbb I\to \Sigma$), then on a local trivialisation, the connection one-form locally is  $\tilde \theta = \theta(t) +\theta_h(t)dt$. Likewise, our choice on the smooth function $\rho$ supported at the neighbourhood of diagonal strata can be encoded in a parameterized version $\tilde \rho$, which is a smooth function
 on $\overline{{\rm Conf}(2,\Sigma)}\times \mathbb{I}$ whose restriction to any $t\in \mathbb I$ is a well-defined cutoff function.

The extended propagator is denoted by $\tilde P(x,y;t) = P(x,y;t)+P_h(x,y;t) dt$ on $\overline{{\rm Conf}(2,\Sigma)}\times \mathbb{I}$ such that 
\[d_\Sigma P = \Xi, \quad  d_\Sigma P_hdt  = \Xi_hdt - d_t P,\] 
following similar procedure to the construction of ordinary propagators. 

The effective action is defined by the computation using the family propagator
$\tilde I_{\rm eff} =I_{\rm eff}(t)+I_h(t)dt\in  \Omega^*_{\mathbb I}\otimes Obs(\Sigma).$
So we have that $I_{\rm eff} = {\rm log}(e^{\hslash\partial_{ P}} e^{\frac{I}{\hslash}})$ and \[I_h = {\rm log}(\hslash \partial_{P_hdt} e^{\hslash\partial_{ P}} e^{\frac{I}{\hslash}}).\]
Graphically $I_hdt$ is obtained effectively by replacing one of the propagator in $I_{\rm eff}$ by $P_h dt$.

Similar to the previous case, we have that
\begin{eqnarray*} Q\tilde I_{\rm eff}[\phi] &=& \sum_{\Gamma} \frac1{|{\rm Aut}\,(\Gamma)|} \int_{\partial_\delta \overline{{\rm Conf}(\Gamma,\Sigma)}} \prod_{e\in {\rm E}(\Gamma)} \hslash \partial^e_{\tilde P} \cdot \prod_{v\in V(\Gamma)} \frac{I_v[\phi]}\hslash \\
&-& \sum_{\Gamma} \frac1{|{\rm Aut}\,(\Gamma)|} \sum_{g\in {\rm E}(\Gamma)} \int_{\overline{{\rm Conf}(\Gamma,\Sigma)}} \hslash \partial_{d_\Sigma \tilde P}^g \cdot\prod_{e\in {\rm E}(\Gamma)\backslash \{g\}} (\hslash \partial^e_{\tilde P}) \cdot  \prod_{v\in V(\Gamma)} \frac{I_v[\phi]}\hslash .
\end{eqnarray*}
The potential obstruction to the family QME is the boundary contribution, which can be shown to vanish parallelising the analysis in the ordinary QME.

Now with the family version of QME proved, whose $\Omega_{\mathbb I}^0$-component gives
\[Q I_{\rm eff}+\hslash \Delta I_{\rm eff}+\frac12\{I_{\rm eff}, I_{\rm eff}\}=0.\]
The $\Omega_{\mathbb I}^1$-component gives
\[ Q I_h +\hslash \Delta I_h+\{I_{\rm eff}, I_h\}+\hslash \Delta_h I_{\rm eff}+\frac12\{I_{\rm eff}, I_{\rm eff}\}_h = d_t I_{\rm eff},\]
where $\Delta_h$ and $\{-,-\}_h$ are $dt$-component of operations $\tilde \Delta$ and $\{-,-\}\tilde{}$ respectively.

The first three terms makes a parameterised BV coboundary. 

The expression $\hslash \Delta_h I_{\rm eff}+\frac12\{I_{\rm eff}, I_{\rm eff}\}_h$ graphically comes from replacing a propagator in $I_{\rm eff}$ by $\Xi_hdt$, which provides an infinitesimal homotopy for deformation of $\Xi$ along $t\in \mathbb I$, i.e., $[Q, \hslash \partial_{\Xi_hdt} ] = \hbar \partial_{d_t\Xi}$.

\end{proof}

The family QME does not quite express the deformation $d_t I_{\rm eff}$ as a quantum BV exact term, since part of the gauge choices determines the BV differential(, recall that the BV operator $\Delta$ and the anti-bracket $\{-,-\}$ are defined via the diagonal class $\Xi$). 
If one compare to the family version of QME in \cite{CG}, the above version is more general in the sense that in \cite{CG}, the author considered the choices in propagator, but suppose that there is no change in the heat kernel side. In our case, the heat kernel side is given by the diagonal class $\Xi$, which further gets promoted to family version due to different choices of cohomology representatives in $\Sigma$. On the other hand, one thing that does keep track of the parameterised $\Xi$  is the renormalisation flow \cite{Co08}, which relates theories with different energy scale (a.k.a, different quantum BV theories). So the $t$-flow $d_t I_{\rm eff}$ should be viewed as an analog of the compatibility between the renormalisation flow equation, and the QME. 
The expression for $d_t I_{\rm eff}$ guarantees that there exists an canonical quasi-isomorphism between the quantum global observables with different gauge choices. In this sense, fixing the gauge choices are similar to specifying coordinates, yielding quasi-isomorphic computation at cohomology level.

\section{The structure of observables}
\label{secobserv}

\subsection{Global observables}
\label{obs:glo}

Vacuum graphs are Feynman graphs that have no external lines. In particular, we look at connected graphs which contribute to the constant term of the effective action, viewed as a functional over the space of field. These are Feynman graphs with vertices coming from classical action and without external legs, which formally is generated by
\[ {\rm log}( e^{\hslash\partial_P} e^{I[\phi]/\hslash})|_{\phi=0}\]
in powers of $\hslash$. 
Since this does not depend on specific choices of differential structures on $\Sigma$, this could be a candidates for topological invariant on $\Sigma$ and a Poisson invariant over $M$.

\begin{proposition}
The total vacuum graph contribution is a topological invariant with respect to the worldsheet topology. If the Poisson structure is regular (including the trivial case), then the vacuum digram contribution vanishes up to homotopy.

\end{proposition}

\begin{proof}
We have shown that the dependence of effective action on the gauge choice of propagator(, and hence the differential structure of $\Sigma$) results in cohomologically trivial terms in the quantum observables $Obs^q$. 
This applies to the vacuum graphs. Those graphs thus become topological invariants of the worldsheet $\Sigma$. 

For any graph $\Gamma$ that contributes, the forms degree of the integrand must be equal to twice of the vertices number, i.e., $E(\Gamma) = 2 V(\Gamma)$. Each edge comes from the contraction of an $\eta$-field and a $X$-field. Since each vertices from the action can at most contribute $2$ $\eta$-fields, 
\[V(\Gamma)\leq E(\Gamma)\leq2V(\Gamma).\] 
Now the second equator holds only when each vertex comes from the $\Pi$-term. In that case, the counting of $X$-fields states 
\[\sum_{v\in V(\Gamma)} {\rm deg}(v) - 2|V(\Gamma)| = |E(\Gamma)|,\] hence the total degree of the graph is four times the number of vertices. Graphically we shall use out-going half-edge to denote $\eta$-field, then the vacuum graphs are generated by graphs, whose vertices all have two out-going edges. And at least one $\Pi$-vertex have one or more in-coming edges. In case of regular Poisson, the Poisson bivector field $\Pi$ is constant. So there exists a choice of coordinate charts on $M$ such that $\Pi$ exhibits constant component. Then the theory does not have vertices with two $\eta$ fields and one or more $X$ fields, which can only be induced by non-constant part of $\Pi$.

So far out discussion depends on a choice of coordinate charts over the Poisson manifold $M$. The choice of coordinate charts determines a connection on $T_M$, and thereby choose a particular $l_\infty$ structure over $\mathfrak g$. However, difference choices of connection on $T_M$ results in contractible homotopies, which does not change our conclusion up to a coboundary term in $Obs$. 
\end{proof}

For the non-vacuum graphs, if the Poisson structure is trivial, the theory is a cotangent theory, where there are only one-loop graphs. In some cases (typically the worldsheet geometry is fixed) there are result showing that those one-loops computes geometric genus associated to the target geometry \cite{Co11, GG_ahat}. In our theory, even with the presence of higher loops, it still makes sense to single out one loops and probe their meaning. Indeed, if we assign a scaling symmetry on $\hslash$, then one-loop graphs give the scaling-invariant contribution. It would be interesting to see what does these graphs tell us about Poisson geometry. Indeed, for all theories of classical BV type, where the space of field is a shifted $l_\infty$ algebra with a symmetric structure, the one loop graphs can be interpreted as certain components of the corresponding PROP of the Lie structure. This enables us to relate the one loop result to certain classes in the Gelfand-Fuks cohomology of the $l_\infty$ algebra. 

\subsection{(Non-)Local observables: factorisation product}
\label{obs:loc}

For any function over the target $B\mathfrak h$ and any point in the boundary of the worldsheet $x\in \partial\Sigma$, there is a corresponding classical observable supported at the point:
\[r^\partial_x:  \mathcal O_{B\mathfrak h}\to Obs_U, x\in U\cap \partial \Sigma\]
which is induced by the evaluation map of fields $\delta_x: \Omega_\Sigma\otimes \mathfrak h\to \mathfrak h: \alpha\otimes X\mapsto \alpha(x)X$. When both the open $U$ and its intersection with the boundary $U\cap \partial\Sigma$ contractible, then local constancy of local observables can be checked since 
\[\mathfrak h\to \Omega_U\otimes \mathfrak h\oplus \Omega_{U, U\cap\partial\Sigma}\otimes \mathfrak h^\vee[-1] \] is an equivalence of $l_\infty$ algebras\footnote{Consider the inclusion of ideal map $\Omega_{U, U\cap\partial\Sigma}\otimes \mathfrak g\to \Omega_U\otimes \mathfrak g$.}. This 
 guarantees that the map $r^\partial_x$ is a weak equivalence. 
Similarly, there is a weak equivalence
\[r_z: \mathcal O_{B\mathfrak g}\to Obs_U, z\in U\subseteq \Sigma^\circ\]
for any bulk point $z$ and for its contractible small neighbourhood $U$.

In the Chern-Simons type of model we considered, the target $B\mathfrak g$ is a formal stack encoding the polyvector field $PV_M$ together with the differential defined by Poisson structure,  which has the structure of a Gerstenharber algebra (a.k.a, a model of $P_2$ algebra). This induces the homotopy $P_0$ structure on local classical observables over any contractible open subset of the inner worldsheet $\Sigma^\circ$. 
It has been shown that for an extended $2$-dimensional TFT, the local observables form a $\mathbb E_2$ algebra \cite{CG, Schei14}.

The product map is given by the quantization map
\[ {\rm Conf}(2,U) \times C^*(\mathfrak g)\otimes C^*(\mathfrak g)\to Obs^q\to C^*(\mathfrak g)[[\hslash]]\]
\[ ((z,w), f,g)\mapsto {\rm log}(e^{\hslash \partial_P}e^{I/\hslash}r_z(f) r_w(g))\mapsto \sum_{\Gamma(z,w)} c_{\Gamma(z,w)} B^{\Gamma(z,w)}(f,g).\]
The coefficient $c_{\Gamma(z,w)}$ comes from integration of admissible graph $\Gamma$ with two given bulk vertices, and $B^{\Gamma(z,w)}$ is the bidifferential operator defined from $\Gamma$.
The last arrow is given by restricting to vacuum graph limit (i.e., consider vacuum graphs which have no insertion of background fields, or external legs). The analytic part gives a number associated to the graph integration, while the algebraic part gives rise to a bi-differential operator $B^{\Gamma(z,w)}$ on $C^*(\mathfrak g)[[\hslash]]$, which then controls the $\mathbb E_2$ multiplication \footnote{Note that due to the failure of the locality even at propagator level, the restriction map does not land in local observables, but comes from the evaluation map of observables on the space of fields. Due to this reason, we added ``(non-)" in front of local in the subtitle of this section.}. The summation of the final term is over all connected graphs with two special separated bulk points $z,w$, on which the corresponding vertices are $f$ and $g$ respectively. 

To simplify the notation, in the remaining part of the paper, we shall use $W(P, f(z)g(w))$ to denote ${\rm log}(e^{\hslash \partial_P}e^{I/\hslash}r_z(f) r_w(g))$. Following the working definition of local observables from \cite{CG}, a quantum observable $O$ is local over an open $U\subset \Sigma$ if there exists a choice of propagators such that $W(P, O)$ is supported over $U^{\times n}$ for certain $n\geq1$. If by construction, we could choose the diagonal class $\Xi$ to be supported within any given neighbourhood of the diagonal ideal, then it is obvious that $W(P, f(z)g(w))\in Obs^q(U)$. However, we do not assume such choice to exists. Also, even in the flat case \cite{Kont97}, the propagator does not necessarily obey the locality condition. So we shall not stick to the local observables in the process of defining the structure. Next, we have Thm.~\ref{obser-prop}. We shall give the prove by the end of this section.

Further more, this also gives a similar map for boundary observables
\[{\rm Conf}(2,\partial U) \times C^*(\mathfrak h)\otimes C^*(\mathfrak h)\to Obs_U\to C^*(\mathfrak h)[[\hslash]]\]
\[ ((x,y), f,g)\mapsto W(P, f(x)g(y))\mapsto\sum_{\Gamma(x,y)} c_{\Gamma(x,y)} B^{\Gamma(x,y)}(f,g).\] 
In the boundary case, ${\rm Conf}(2,\partial U)$ has two connected components, corresponding to the associative multiplication of deformation quantization and its dual. Note that this is a generalisation of Kontsevich's construction of star product --- Kontsevich defined the product in the case where the worldsheet is flat, but here the multiplication comes from integrations over configuration space of Riemann surfaces. The summation of the final term is over all connected vacuum graphs with two special separated boundary points $x,y$. 
Combining the bulk and boundary multiplication, this makes $(C^*(\mathfrak g), C^*(\mathfrak h))$ into a Swiss-Cheese algebra in the sense of \cite{Vor97}.

\begin{proof} We shall check the dependence of the correlation on the insertion points, and on our choices of propagators, and hence finish the proof of Prop.~\ref{obser-prop}. So as before, we extend the action and the propagator to the parameterised case, i.e., we consider the correlation on ${\rm Conf}(\{x,y\}, \partial U)\times \mathbb{I}$. 
\begin{eqnarray*}d W(\tilde P, f(x) g(y))[\phi] &=& {\rm boundary\,contributions}+\\
&& W(\tilde P, d_\Sigma\tilde P, f(x) g(y)) [\phi]+W(\tilde P, f(x)g(y))[\phi; Q\phi].
\end{eqnarray*}
The new notation $W(\tilde P, d_\Sigma\tilde P, f(x) g(y))$ corresponds to a summation of similar Feynman graph contribution where one propagator $\tilde P$ is replaced by the parameterised diagonal class $d_\Sigma\tilde P$.

The boundary contribution could come from integrating over two individual strata. Firstly, whenever we send any bulk point of a graph (not the boundary insertion points) to the boundary, the integration vanishes. Secondly, whenever two or more bulk points collide, this corresponds to a collapsing of subgraph, which we can view as vacuum graph. It's vanishing is due to a similar argument of Kontsevich's vanishing result. 

As before, we denote by $\tilde \Delta$ and $\{-,-\}\tilde{}$ the BV operator and anti-bracket defined using the parameterised diagonal class $\tilde \Xi$. Now the second term $W(\tilde P, d_\Sigma\tilde P, f(x) g(y)) [\phi]$ is given by Feynman graphs with two fixed boundary insertion points and one $\tilde \Xi-d_t P$-edge. The $\tilde \Xi$-edge graphs correspond to two expressions, depending on whether this $\tilde \Xi$ edge is separating edge or not. In the former case, the graphs contribute to $\{\tilde{I}_{\rm eff}, W(\tilde P, f(x) g(y))\}\tilde{}$. In the latter case, the graphs correspond to $\tilde \Delta W(\tilde P, f(x) g(y))$. The $d_t P$-edge case gives the $t$-flow term for the quantum observable. 

The third term $W(\tilde P, f(x)g(y))[\phi; Q\phi]$ corresponds to the graph integration where one external leg, which inserts field $\phi$, is replaced by the insertion of $Q\phi$. By definition, this is $QW(\tilde P, f(x)g(y))[\phi]$.  

So we have that 
\begin{eqnarray*} 
(d+d_t) W(\tilde P, f(x)g(y))  &=& QW(\tilde P, f(x)g(y)) +\hslash\tilde \Delta W(\tilde P, f(x)g(y))\\
&&  +\{I_{\rm eff}, W(\tilde P, f(x)g(y))\}\tilde{},\end{eqnarray*}
similar to the proof of parameterised QME.
At the $\Omega^0_{\mathbb I}$-component, the homotopy equation gives that
 \begin{eqnarray*} 
d W( P, f(x)g(y)) &=& QW( P, f(x)g(y))  +\hslash \Delta W( P, f(x)g(y)) \\
&&+\{I_{\rm eff}, W( P, f(x)g(y))\}.\end{eqnarray*}
The right hand side is a BV coboundary. So up to BV cohomology, the parameterised multiplication has no local dependence on the configuration space ${\rm Conf}(\{x,y\}, U\cap \partial \Sigma)$.

At the $\Omega^1_{\mathbb I}$-component, the homotopy equation reads
\begin{eqnarray*} 
d_t W( P, f(x)g(y))&=&- d W(P, P_hdt, f(x)g(y)) [\phi]  \\
 &&+QW(P, P_hdt, f(x)g(y)) [\phi]+\hslash \Delta W(P, P_hdt, f(x)g(y)) \\
 &&+\{I_{\rm eff}, W(P,P_hdt, f(x)g(y))\}\\
 &&+ \hslash  \Delta_h W(P,  f(x)g(y)) +\{I_{\rm eff}, W( P, f(x)g(y))\}_h\\
 &&+\{I_h, W( P, f(x)g(y))\}.
\end{eqnarray*}
Since we only care about the gauge dependence at this level, one can ignore the first line at the right hand side. The second to third lines give a BV exact term, while the last two lines show the consistency among different gauge choices, whose integral form allows one to identify computations using different gauge choices.
\end{proof}

\subsection{On boundary observables}
\label{obs:bdy}

The product structure for boundary observables gives an associative algebra structure on $C^*(\mathfrak h)$, which is a formal geometry avatar for the smooth algebra 
\[{\rm Conf}(2, \partial \Sigma)\times C^\infty(M)\times C^\infty(M)\to C^\infty(M)[[\hslash]].\]

It interesting to look into relevant Feynman graphs.
Firstly, there is no tadpole contribution. Secondly, as we previously explained, there can not be external $\eta$-legs.
There are the $l_n$-vertices coming from the cotangent theory sector, as well as the $\Pi$-vertices. This gives rise to the admissible graphs as in Kontsevich defined \cite{Kont97} corrected due to the nontrivial target geometry and the non-constancy of Poisson bivector field. Each summand is a bidifferential operator acting on those two functions multiplied by a graph evaluation, which is an integration of products of propagators over the (compactified) configuration space of $n$ points. 

Given an $l_\infty$ algebra $(\mathfrak h,\{-,-\}_0)$ over $\Omega_M$ equipped with a symplectic pairing, there is a one-dimensional field theory given by the following data:
\begin{enumerate}
\item The space of field is given by $\Omega_{S^1}\otimes_{\mathbb R} \mathfrak h[-1]$, which is an $l_\infty$  $-1$-symplectic space. The shifted symplectic structure is given by the pairing among forms over $S^1$ and the ordinary symplectic structure over $\mathfrak g$.
\item The classical action is given by $\int_{S^1}\langle X, \frac12 dX+\sum_{n\geq0}\frac{1}{(n+1)!} l_n(X)\rangle_0$
\item The classical observables are weakly equivalent to 
\[\widehat{\rm Sym} \,\left(\Omega_{S^1}\otimes_{\mathbb R} \mathfrak h^\vee[1]\right)\,, d_{S^1}+d^{CE}\]
by Atiyah-Bott lemma.
\item The BV quantization is unobstructed. 
\end{enumerate}

The quantization of the aforementioned theory has been studied by \cite{GG_ahat,GLL}. We shall show that this theory is a boundary theory of Poisson sigma model in the symplectic case. To achieve this, we start with a symplectic manifold $M$, and choose the $l_\infty$ structure to be one which is compatible with the symplectic form. This can be done when we choose the symplectic connection to be the $l_1$ structure on $\mathfrak h$. This uniquely fixes the $l_\infty$ structure, and guarantees that the symplectic pairing is invariant, see Appendix for details.  
One byproduct of such choice is that now the symplectic form(, and hence the Poisson bivector) is constant along the parallel coordinates defined by the connection. Hence all the $\Pi$-vertices corresponding to the Poisson structure is simplified to a bilinear term only. Namely we have that 
\[S^{BV}=\int_\Sigma\langle \eta, dX+\sum_{n\geq0}\frac1{n!}l_n(X)+\frac12\Pi(\eta)\rangle_{1}\]
for $X\in \Omega_{\Sigma}\otimes \mathfrak h$ and $\eta\in \Omega_{\Sigma}\otimes \mathfrak h^\vee[-1]$.
The differential in the classical observables hence comes from two parts: that from cotangent theory, and the extra part of $l_1$ structure from $\Pi$. The differential \[d^K X(x)=\Pi (\eta(x))+dX(x)+\sum_{n\geq0}\frac1{n!} l_n(X^{\otimes n}(x))\]
can be solved ``roughly'' by inverting $\Pi$. 
The Maurer-Cartan element satisfies the expression 
\[\Pi (\eta(x))+dX(x)+\sum_{n\geq0}\frac1{n!} l_n(X^{\otimes n}(x))=0,\]
which behaves differently when we take the worldsheet point $x$ to be at the boundary and at the bulk. If $x\in \partial \Sigma$, by our choice of boundary condition we have that 
\[dX(x)+\sum_{n\geq0}\frac1{n!} l_n(X^{\otimes n}(x))=0\]
and this identifies $X$ with a boundary Maurer-Cartan element of $\Omega_{\partial\Sigma} \otimes\mathfrak h$. If $x\in \Sigma\backslash\partial \Sigma$,
the equation says that the solution for $\eta$ is given by the obstruction to a bulk Maurer-Cartan element for any bulk extension of a boundary Maurer-Cartan element. This clearly shows that the theory has only interesting geometry at the boundary observables.

The action, upon solving the equation of motion for $\eta$, is given by the restriction on the subspace spanned by $$\left(\eta(x) +\Pi^{-1} [d X(x)+\sum_{n\geq0}\frac1{n!} l_n(X^{\otimes n}(x))]\right).$$ So the action is given by
\begin{eqnarray*}S[X]&=&-\int_\Sigma  \frac12\langle \Pi^{-1} d  X+\sum_{n\geq0}\frac1{n!} \Pi^{-1}_\epsilon l_n( X^{\otimes n}), \Pi^{-1} d X+\sum_{n\geq0}\frac1{n!} \Pi^{-1} l_n( X^{\otimes n})\rangle_1 
.\end{eqnarray*}

The action is a total derivative which, upon restricting to the boundary, gives the desired 1d action.
\[S^{\rm 1d}[X] = \int_{\partial\Sigma} \langle X, \frac12dX+\sum_{n\geq0}\frac1{(n+1)!} l_n(X^{\otimes n})\rangle_0.\]
Terms such as \[\int_\Sigma \frac12\langle \frac1{n!}\Pi^{-1} l_n( X^{\otimes n}), \frac1{m!}\Pi^{-1} l_m( X^{\otimes m})\rangle_1\] combines to a constant functional as follows.
Consider the variation of such term, which leads to 
 \[\frac12 \int_\Sigma \langle \frac1{(n-1)!} l_n( \delta X, X^{\otimes n-1}), \frac1{m!} l_m( X^{\otimes m})\rangle_0 + \langle \frac1{n!} l_n( X^{\otimes n}), \frac1{(m-1)!} l_m( \delta X, X^{\otimes m-1})\rangle_0.\]
The symplectic pairing $\langle-,-\rangle_0$ is again invariant with respect to the $\{l_n\}$ operations, so can be rearranged into
\[\frac12\sum_{m+n}\int_\Sigma \frac{|{\rm shuffle}(m,n-1)|}{(n+m-1)!}\langle  l_n( l_m( X^{\otimes m}) , X^{\otimes n-1}), \delta X\rangle_0 \] and \[\frac12\sum_{m+n} \int_\Sigma \frac{|{\rm shuffle}(m-1,n)|}{(n+m-1)!}\langle \delta X, l_m( l_n( X^{\otimes n}) , X^{\otimes m-1})\rangle_0.\]
Those two summations each cancel among themselves due to the constraints on the $l_\infty$ structure
\footnote{Technically the action is a functional valued in cdga $(\Omega_M, d_M)$, and we drop $d_M$-exact terms.}. 

For the future convenience,  we adjust the one-dimensional boundary action a little bit by adding in a non-interacting sector
\[S^{\rm 1d}[X, \eta] = \int_{\partial\Sigma} \langle X, \frac12dX+\sum_{n\geq0}\frac1{(n+1)!} l_n(X^{\otimes n})\rangle_0+\frac12\int_\Sigma \langle \eta,\Pi(\eta)\rangle_1.\qquad (\dagger)\]
One can see easily that from the boundary point of view, with the extra term, nothing sensitive could be changed.
However, the previous analysis can not be considered as a direct proof in the $l_\infty$ algebra setting. The homotopy $\Pi^{-1} : \Omega\otimes \mathfrak g^\vee\to \Omega_D\otimes \mathfrak g[1]$ is not well-defined, since its image should be in Dirichlet forms. This technical issue is by-passed by the following proposition, and where we only understand $\Pi^{-1}$ as a pairing at the target $l_\infty$ algebra $\mathfrak h$.
\begin{proposition}
The BV action for Poisson sigma model and the 1d action given by $(\dagger)$ are related by a homotopy $S+\lambda \{H, S\}$ where \[H=-\int_{\Sigma}\Pi^{-1}(X, \frac12dX+\sum_{n\geq0}\frac1{(n+1)!} l_n(X^{\otimes n})).\]
\label{classicalbdythy}
\end{proposition}

\proof Firstly one notes that $H$ is a well-defined classical observable, then the computation is straight-forward.\qed

In the case where the Poisson target is not symplectic, we do not expect the theory to be fully reduced to boundary. In the case of regular Poisson, along the direction where the Poisson bivector degenerates, the theory behaves like a cotangent theory, which has no obvious reason to exhibit holographic properties.

\subsection{The existence of a quantised boundary theory}

We have seen that the Poisson model, in the symplectic case, is equivalent to a boundary theory classically. Now we shall give the proof of Thm.~\ref{bulkbdy}, which provides a homotopy for $I_{\rm eff}$ and $S^{1d}_{\rm eff}$.

\begin{proof}
By definition, the quantum corrected functional $H_{\rm eff}$ is given by
\[\sum_\Gamma \frac1{|{\rm Aut}\,(\Gamma)|} \int_{\overline{{\rm Conf}(\Gamma, \Sigma)}} 
(\prod_{e\in IE(\Gamma)} \hslash \partial^e_{P}) \, (\prod_{v\in V(\Gamma)\backslash\{v_0\}} \hslash^{-1}I_v )\cdot H_{v_0},\]
where the summation is over finite directed graphs with a fixed base point $v_0$.
Next we compute $dH_{\rm eff} + \hslash\Delta H_{\rm eff} +\{I_{\rm eff}, H_{\rm eff}\}$.

The above expression graphically corresponds to contributions labeled by connected admissible graphs. The first term comes from admissible graph with one $H$-vertex and finite $I$-vertices, where the operator $d$ acts on external half-edges. The second (resp. the third) term is the admissible graphs with an inner edge corresponding to $dP$ in the loop (resp. connecting two components, see for example Fig.~\ref{poisson_4_1}). 

\begin{figure}
\begin{center}
\includegraphics[width=3in]{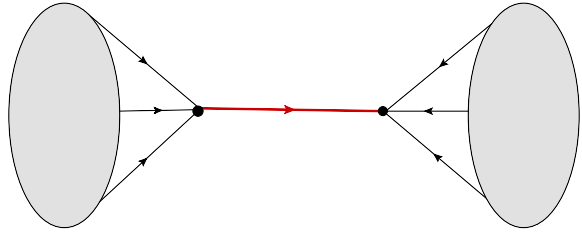}
\end{center}
\caption{\footnotesize An illustration of graphs contributing to $\{I_{\rm eff}, H_{\rm eff}\}$, where the red edge corresponds to $dP$, the shaded bubbles denote the connected subgraph components.}
\label{poisson_4_1}
\end{figure}

The sum corresponds to the integration of exact terms over configuration spaces, and hence can be computed by restricting the graphs to the boundary strata of $\overline{{\rm Conf}(\Gamma, \Sigma)}$.

It is informative to probe into the formations of boundary strata --- they fall into the following four cases.

\begin{itemize}
\item Some bulk vertices coming from $I$ collapse at the boundary. 
\item Some bulk vertices (could be none) coming from $I$ and one $H$-vertex collapse at the boundary, see Fig.~\ref{poisson_4_4}. 
\item Three or more bulk vertices collide into one single bulk vertex. 
\item Two bulk vertices collide into one single bulk vertex. \end{itemize}

The first case does not contribute due to our choice of boundary conditions. Nor does the third case, since this contradicts the vanishing result as shown in proving QME.

In the second case, due to the same reason, the only case that can survive is given by restricting a single $H$-vertex to the boundary. The $H$-vertices at boundary are in one-one correspondence to the vertices in $S^{1d}$, and hence each such graph evaluates the quantum corrected $S^{1d}$.
For example, in Fig.~\ref{poisson_4_4}, restricting the $H$-vertex $\int_{\Sigma}\langle X, \frac1{3!}l_2(X^{\otimes 2})\rangle_0$ to the boundary is equivalent to replacing the $H$-vertex by $\int_{\partial\Sigma}\langle X, \frac1{3!}l_2(X^{\otimes 2})\rangle_0$ in the given graph, which contributes to the quantum correction to the latter.

\begin{figure}
\begin{center}
\includegraphics[width=2.5in]{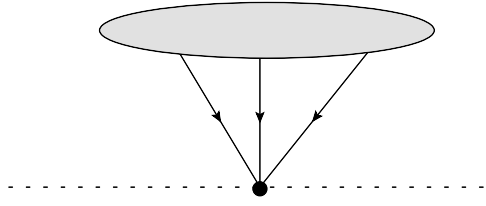}
\end{center}
\caption{\footnotesize An illustration of quantum corrections of $H$, where $H$-vertex is restricted to the boundary, the shaded bubble denotes the rest part of the graph.}
\label{poisson_4_4}
\end{figure} 

The last case concerns the shrinking of a single edge due to dimension reasons. So this can be further classified into several sub-cases.
\begin{enumerate}
\item The edge connects one $l_n$-vertex to another $l_m$-vertex.
\item The shrinking edge connects one $l_n$-vertex to the $\Pi$-vertex.
\item The shrinking edge connects one $l_n$-vertex to an $H$-vertex.
\item The shrinking edge connects one $\Pi$-vertex to an $H$-vertex.
\end{enumerate}
The first two sub-cases vanish due to the same reason as in proving QME. 
The third sub-case vanishes since it produces effectively a vertex proportional to $l_n\circ l_m$, which vanishes upon summation over $m+n$.
In the forth sub-case, suppose the $H$-vertex is given by $-\int_\Sigma \langle X, \frac1{(n+1)!} l_n(X^{\otimes n})\rangle_0$, then the quotient graph upon the collapsing  involves effectively an $l_n$-vertex. So those gives rise to a graph contains an $l_n$-term in $I_{\rm eff}$ up to a sign. The analytic part involves integration of one propagator over the diagonal boundary strata of the two-point configuration space, which gives a unit.

On the other hand, this is nothing but the quantum correction to a boundary observable $\int_{\partial\Sigma}\Pi^{-1}(X, \frac1{4!}l_3(X^{\otimes 3}))$. 
So we have that \[dH_{\rm eff}+\Delta H_{\rm eff}+\{I_{\rm eff}, H_{\rm eff}\} = S^{1d}_{\rm eff} - I_{\rm eff},\]
which completes the proof.
\end{proof}

\section{Outlooks}

In the current paper, we investigated the classical and quantum BV theory originated from Poisson sigma model via formal geometry approach. As a topological theory, the theory was formulated on two-dimensional Riemann surfaces with boundaries, which generalises the flat case as in \cite{Kont97}. 

Since the formulation of quantum BV theories in terms of a quantum derived moduli problem \cite{Co11}, not many theory with boundaries have been studies. On the other hand, the boundary theory we are dealing with in the current paper exhibits a quite universal pattern. Namely, Kontsevich's deformation quantization used the polyvector fields over Poisson manifold as the classical object to be quantised, as opposed to the Poisson manifold itself. The polyvector construction has now been understood as a shifted Poisson variety which retains many properties of the original Poisson structure. The difficulty in formulating a Lagrangian theory for Poisson model in one-dimension is just a specific example. A general construction of ``center" in the derived Poisson world shall enable us to mimic Kontsevich's construction in more cases. 

If one takes the polyvector construction of some perhaps shifted symplectic structure and consider the corresponding BV theory as we did here, then it is immediate that such theory is trivial over a manifold without boundaries classically. So all the potential interesting phenomena happens at the boundary. 

BV quantum theory over a flat disk is expected to gives rise to an observable algebra over the Swiss-Cheese operad. In the general case, we looked into some observables and identified the SC-module structure. However, locality is not guaranteed. Even from the construction of the propagator, the cohomology representatives of the worldsheet need not be local. It is interesting to further investigate the cases where locality is absent in Lagrangian field theories, which might shed some light on possible formulations of non-perturbative theories. 

\section*{Appendix: The invariant pairing for symplectic connection}

It is known that a choice of symplectic connection determines an $l_\infty$ structure \cite{Fedosov, GLL}. In this appendix we shall show that the symplectic pairing on the so-defined $l_\infty$ algebra is invariant.

To show that the symplectic pairing is invariant, we need to look back into the procedure of solving the dual $l_\infty$ structure recursively from the flatness condition of the Abelian connection of the underlying Weyl bundle, as shown in the references above. 

Indeed, the Chevalley-Eilenberg differential in $C^*(\mathfrak h)$ can be decoupled into two parts, \[\partial = \delta + (\sum_{n\geq 1} l_n^\vee(e^i))\partial_{e^i},\quad \delta=l_0^\vee(e^i)\partial_{e^i},\] where $\{e^i\}$ is a chosen $\Omega$-linear basis of $\mathfrak h^\vee[-1] \equiv \Omega\otimes \mathcal T^\vee$. The $l_0$ part itself is a derivation which is a differential, and if the basis of $\Omega\otimes \mathcal T^\vee$ is compatible with the local coordinates, we have that $l_0^\vee(e^i)\partial_{e^i} = dx^i\partial_{e^i} $. Moreover, it has a homotopy contraction \footnote{The contraction $\delta^{-1}$ is a cochain map, which does not preserve the algebra structure.} which we shall denote by $\delta^{*}$ and is locally 
\[\frac1{p+q} y^i\iota_{\frac\partial{\partial x^i}}\quad, {\rm on}\, \Omega^p\otimes {\rm Sym}^q(\mathcal T^\vee).\] It is easy to check that 
\[\delta\circ \delta^{*}+ \delta^{*}\delta = id \,\,{\rm on}\, \oplus_{p+q>0}\Omega^p\otimes {\rm Sym}^q(\mathcal T^\vee).\]

The nilpotence of CE differential results in the relation on $l_n$ by restricting the image component-wise. So from the relation
\[\delta \circ l_n^\vee (e^i)\partial_{e^i}  = \sum_{a+b=n; 0<a,b<n} l_a^\vee(e^i)\partial_{e^i}\circ l_b ^\vee(e_j) \partial_{e^j},\]
one can obtain a unique solution of $ \{l_n^\vee\}_{n\geq2}$ recursively once we determine $l_1$ and impose the gauge fixing condition that $\delta^*\circ l_n^\vee (e^i)\partial_{e^i} = 0$. In fact, due to the previous prescription, we have that 
\[l_n^\vee (e^i)\partial_{e^i}  =  \delta^{*}\circ l_1^\vee(e^i)\partial_{e^i}\circ l_{n-1} ^\vee(e_j) \partial_{e^j}= \{\delta^{*}, l_1^\vee(e^i)\partial_{e^i}\}\circ l_{n-1} ^\vee(e_j) \partial_{e^j},\]
for $n\geq3$ and
\[l_2^\vee (e^i)\partial_{e^i}  =  \delta^{*}\circ l_1^\vee(e^i)\partial_{e^i}\circ l_{1} ^\vee(e_j) \partial_{e^j}.\] 
Assuming skew-selfadjointness, the $l_1$ can be identified with a symplectic connection $\nabla$ on $\Omega\otimes\mathcal T_M$(, and correspondingly on $\Omega\otimes\mathcal T^\vee_M$). Now viewing $\delta^*\circ l_1^\vee(e^i)\partial_{e^i}$ as a linear map on ${\rm Sym}^*\mathfrak h^\vee[-1]\equiv {\rm Sym}^* (\Omega\otimes\mathcal T^\vee_M)$, its action on each basis is defined by 
\begin{eqnarray*}f_{i_1\cdots i_k} e^{i_1}\cdots e^{i_k} &\mapsto& \delta^* (\sum_{j} (l_1^\vee)_{m}^{i_j} f_{i_1\cdots i_j\cdots i_k} e^{i_1}\cdots  e^m \hat{e^{i_j}} \cdots e^{i_k} + df_{i_1\cdots i_k} e^{i_1}\cdots e^{i_k})
 \\&\sim&
 \sum_{j}\iota_{\frac\partial{\partial x^n}} (l_1^\vee)_{m}^{i_j} f_{i_1\cdots i_j\cdots i_k}  e^{i_1}\cdots e^n e^m \hat{e^{i_j}} \cdots e^{i_k} + \partial_n f_{i_1\cdots i_k} e^n e^{i_1}\cdots e^{i_k}
\end{eqnarray*}
if we assume that 
\[\delta^*f_{i_1\cdots i_k} e^{i_1}\cdots e^{i_k}=0. \]
Dually, when $l_1$ is given by a symplectic connection, we have that 
\[
l_{n+1} (e_{i_0}, \cdots, e_{i_n})= \sum_{a=0}^{n} \frac1{n+1} (\nabla_{i_a}l_{n})(e_{i_0}, \cdots, \hat{e_{i_a}}, \cdots e_{i_n})
\]
for $n\geq2$ and
\[ l_{2} (e_m, e_n)^\bullet=\iota_{\frac\partial{\partial x^m}} R ^\bullet_n = [\iota_{\frac\partial{\partial x^m}}, (\nabla^2)^\bullet_n ].\]

\begin{proposition}
The symplectic pairing $\langle-,-\rangle_0$ is symmetric with respect to $\{l_n\}_{n\geq1}$ constructed as above if $l_1$ is given by a symplectic connection $\nabla$. 
\end{proposition}

\begin{proof}
Recall that $\langle v_0, l_1(v_1)\rangle_0 = \langle v_1, l_1(v_0)\rangle_0$ by the skew-selfadjointness of $l_1$ and skew symmetry of the symplectic pairing.
The 2-operation $l_2 = R^\nabla$ viewed as an element in $\Omega^1({\rm Hom}(T\otimes T, T))$. One needs to do the consistency check that $l_2$ is further an element in $\Omega^1({\rm Hom}({\rm Sym}^2 T, T))$, but this is guaranteed by the fact that the symplectic connection is torsion free. Recall that $R^\nabla$ is the curvature 2-form valued in the symplectic Lie algebra $\mathfrak{sp}$, so $\langle l_2(-,u), v\rangle$ is symmetric with respect to $u$ and $v$. This, combined with the symmetry in $R^\nabla$, states that $\langle l_2(u,v),w\rangle$ is totally symmetric. 

Similar, for general operation $l_n=  \nabla ^{n-2} R$, we will show that $l_n\in \Omega^1({\rm Hom}({\rm Sym}^n T, T))$ via induction. In fact, suppose 
\[l_{n-1}\in \Omega^1({\rm Hom}({\rm Sym}^{n-1} T, T)),\,\forall v_1, \cdots, v_n\in \mathfrak g,\] 
then
\begin{eqnarray*}l_n(v_1,\cdots, v_n) &=& \nabla_{v_1} l_{n-1}(v_2, \cdots, v_n)=\nabla_{v_1}\nabla_{v_2} l_{n-2}(v_3, \cdots, v_n) \\
&=& \nabla_{v_2}\nabla_{v_1} l_{n-2}(v_3, \cdots, v_n)=l_n(v_2,v_1,\cdots, v_n).\end{eqnarray*}
The last step is again due to the torsion-free property of the symplectic connection. By induction hypothesis, $v_2,\cdots, v_n$ are symmetric, now with $v_1, v_2$ symmetric, this shows that $l_n\in \Omega^1({\rm Hom}({\rm Sym}^n T, T))$.

In the expression, 
\[\langle v_0, l_{n}(v_1,\cdots, v_n)\rangle \equiv \langle v_0, \nabla_{v_1}\cdots \nabla_{v_{n-2}} R(v_{n-1}, v_n)\rangle,\]
$v_0$ and $v_n$ are symmetric, which comes from the algebraic property of the symplectic Lie algebra $\mathfrak{sp}$. Combining the symmetry property of $l_n$, this completes the proof.
\end{proof}

\bibliographystyle{plain}

\end{document}